\documentclass[a4wide,reqno,11pt]{article}
\usepackage{tikz} 
\usetikzlibrary{shapes,snakes,calc,fit}
\usepackage{amsmath,amssymb,amsthm,mathtools,cancel,amsthm,amssymb,amsmath,mathrsfs,stmaryrd,lscape,rotating,adjustbox,
	mathdots,mathtools,color,graphicx,framed,authblk,bbm,todonotes}
\mathtoolsset{showonlyrefs}
\usepackage[numbers,sort&compress]{natbib}

\usepackage[margin=1in]{geometry}
\usepackage{grffile}
\usepackage{cancel}
\usepackage[colorlinks=true, pdfstartview=FitV, urlcolor=blue, citecolor=red, linkcolor=blue]{hyperref}
\definecolor{shadecolor}{rgb}{0.9, 0.9, 0.86}

\renewcommand{\d}{\mathrm{d} }
\renewcommand{\emph}[1]{\textit{#1}}
\newcommand{\trace}[1]{\textrm{Tr}\left( #1\right)}
\newcommand{\tr}[1]{\textrm{Tr}\left( #1\right)}
\newcommand{\balpha}{\ba}
\newtheorem{conjecture}{Conjecture}
\newcommand{\Volt}{\textrm{Volt}}
\newcommand{\INB}{\textrm{INB}}
\newcommand{\bigzero}[2]{\mbox{\normalfont\Large\bfseries 0}_{{\small#1\times#2}}}
\newcommand{\ba}{\mathbf{a}}

\newcommand{\bx}{\mathbf{x}}
\newcommand{\bp}{\mathbf{p}}
\newcommand{\br}{\mathbf{r}}
\newcommand{\by}{\mathbf{y}}
\newcommand{\di}{\mathrm{d}}
\newcommand{\wt}[1]{\widetilde{#1}}
\newcommand{\R}{\mathbb{R}}
\newcommand{\Z}{\mathbb{Z}}
\newcommand{\N}{\mathbb{N}}
\newcommand{\la}{\left\langle}
\newcommand{\ra}{\right\rangle}
\newcommand{\fL}{\mathfrak{L}}
\newcommand{\C}{\mathbb{C}}
\newcommand{\cC}{\mathcal{C}}
\newcommand{\wo}[1]{\overline{#1}}
\newcommand{\wh}[1]{\widehat{#1}}
\newcommand{\cM}{\mathcal{M}}
\newcommand{\cL}{\mathcal{L}}
\newcommand{\cA}{\mathcal{A}}
\newcommand{\cE}{\mathcal{E}}
\newcommand{\D}{\mathbb{D}}

\def\be{\begin{equation}}
	\def\ee{\end{equation}}

\DeclareMathOperator{\diag}{diag}
\DeclareMathOperator{\rank}{rank}
\DeclareMathOperator{\Div}{div}

\newtheorem{theorem}{Theorem}[section]
\newtheorem{definition}[theorem]{Definition}

\newtheorem{lemma}[theorem]{Lemma}
\newtheorem{remark}[theorem]{Remark}
\newtheorem{proposition}[theorem]{Proposition} 
\newtheorem{corollary}[theorem]{Corollary}

\title{Discrete integrable systems and  random Lax matrices}
\author{T. Grava\footnote{
		International School for Advanced Studies (SISSA), Via Bonomea 265,   34136, Trieste, Italy,  School of Mathematics,  University of Bristol,  Fry Building,  BS8 1UG,   UK,  and  INFN, Sezione di Trieste
\newline
		\textit{Email:} \texttt{grava@sissa.it} 
	}, M. Gisonni\footnote{
		International School for Advanced Studies (SISSA), Via Bonomea 265,   34136, Trieste, Italy \newline
		\textit{Email:} \texttt{mgisonni@sissa.it} 
	}, G. Gubbiotti\footnote{University of Milan, via Cesare Saldini 50, 20133, Milan, Italy, and INFN Sezione di Milano, via G. Celoria 16, 20133 Milan, Italy  \newline
		\textit{Email:} \texttt{giorgio.gubbiotti@unimi.it}},
	G. Mazzuca\footnote{Department of Mathematics, The Royal Institute of Technology, Lindstedts\"agen 25, 114 28, Stockholm, Sweden. \newline
		\textit{Email:} \texttt{mazzuca@kth.se}}}
\date{\today}

\begin{document}
	
	\maketitle
	
	\begin{abstract}
		We  study  properties of  Hamiltonian  integrable
		systems with random initial data by considering their Lax
		representation. Specifically, we investigate the spectral behaviour of
		the  corresponding Lax matrices when the number $N$ of degrees of freedom of the system goes to infinity 
		and the initial data is sampled
		according to a properly chosen Gibbs measure.  We give an exact
		description of the limit density of states for the exponential
		Toda lattice and the Volterra lattice in terms of the Laguerre and antisymmetric Gaussian $\beta$-ensemble in the high temperature regime.
		For generalizations of the Volterra lattice to short range interactions,   called INB
		additive and multiplicative lattices,  the focusing Ablowitz--Ladik lattice and  the focusing Schur flow,   we derive numerically the
		density of states. For all these systems, we obtain explicitly the density of states in the ground states.
	\end{abstract}
	
	\section{Introduction}

	In this manuscript, we  study  properties of  Hamiltonian  integrable systems
	with random initial data by analysing  the spectral properties
	of their Lax matrices considered as \emph{random matrices}. 
	
	One of the first investigations in this direction was made in \cite{Bogomolny2009,Bogomolny2011}, where the authors considered random Lax matrices of various Calogero--Moser systems, and computed their joint eigenvalues densities. 
	More recently, this idea gained new attention thanks to the work of Spohn \cite{Spohn2019a},
	where the author connected the spectrum of the random Lax matrix of the Toda lattice with the sparse matrix of the Gaussian  $\beta$-ensemble
	\cite{Dumitriu2002} at high temperature \cite{Allez2012}.  Later, the
	Lax matrix of the  defocusing Ablowitz--Ladik lattice \cite{Nenciu2005,Simon2005}
	with random initial data was connected with the sparse matrix of
	the circular $\beta$-ensemble \cite{KillipNenciu2004} in the high temperature
	limit \cite{Hardy2020}  by two of the
	present authors \cite{grava2021generalized}, and, independently, by Spohn
	\cite{spohn2021hydrodynamic}. In  the same work, Spohn also considered the defocusing Schur flow \cite{Golinskiui2006}, and he connected it to the Jacobi $\beta$-ensemble \cite{Dumitriu2002} in the high temperature regime \cite{Allez2013}. Further developments on this subject were
	also presented in \cite{mazzuca2022large,guionnet2021large}. We mention  also the work \cite{Yukawa1986}, where the author studied the statistical properties of the energy level of a quantum integrable system analysing the eigenvalues of the Lax operator.
	
	Classically, from the seminal paper of Liouville \cite{Liouville1855}, an
	integrable system is understood as a \emph{Hamiltonian system} possessing
	enough first integrals such that the motion of the system is regular
	and predictable \cite{Arnold1967}.  In the modern geometrical setting
	\cite{Arnold1997,Olver1986}, we say that given a Poisson manifold $\mathcal P$
	of dimension $n$, with local coordinates $\ba =(a_1,\dots, a_n)$, and
	Poisson bracket $\left\{.\, ,. \right\}$ of rank $2r$, the flow generated
	by a scalar   smooth function $H=H(\ba)$:
	\begin{equation}
		\dot{a}_{i} = \frac{\partial a_i}{\partial t} = \left\{ a_{i}, H \right\},
		\quad
		i=1,\dots,n,
		\label{eq:hamflow}
	\end{equation}
	is \emph{integrable} if it possesses $k=n-r$ functionally independent
	first integrals in involution with respect to the given Poisson bracket.
	
	Finding first integrals is often a complicated task,  and during the
	past decades several algorithms to construct them have been
	developed.  One of the  most effective  methods to produce first integrals
	of a given mechanical system is the so-called \emph{Lax pair representation}\footnote{Often $L{\text -}A$ pair in the
		Russian literature.}.  The concept of Lax pair originates from the work
	of P.~D.~Lax on the theory of  PDEs \cite{Lax1968}, see also \cite{Ablowitz1981,Novikov}.
%	 where it was used to produce exact solutions through the so-called inverse
%	scattering method \cite{Ablowitz1981,Ablowitz1991}.
%	
	 In the finite dimensional setting,
	the construction runs as follows \cite{Babelon,Arnold1987book}:
	assume  there  are $N\times N $  matrices, $L=L(\ba)$ and $A=A(\ba)$ with $N=N(n)$ and such that the equation
	%
	%sucdimensional time dependent vector $\psi=\psi(t)$
	%such that it satisfies two equations,
	%\begin{equation}
	%    L(\ba) \psi = \lambda \psi,
	%    \quad
	%    \dot{\psi} = A(\ba) \psi,
	%    \label{eq:psiprobl}
	%\end{equation}
	%where $L=L(\ba)$ and $A=A(\ba)$ are square matrices.  The system
	%\eqref{eq:psiprobl} for $\psi$ is \emph{overdetermined}, and
	%the compatibility condition is given by the matrix relation
	\begin{equation} 
		\label{Lax0} 
		\dot{L}=[A,L],
		\quad
		[A,L]=AL-LA\,,
	\end{equation}
	is equivalent to the Hamiltonian flow \eqref{eq:hamflow}. Then 
	the matrices $L$ and $A$ are a Lax pair for the Hamiltonian system  and the matrix $L$ is called   \emph{Lax matrix}.
	The main consequence of the  Lax equation \eqref{Lax0}   is that the eigenvalues of $L$ 
	are first integrals of the Hamiltonian flow \eqref{eq:hamflow}.
	So, provided that we can prove these eigenvalues give enough
	functionally independent quantities in involution, we can infer the
	integrability of  Hamilton's equations \eqref{eq:hamflow} through the Lax pair.
	
	The fact that in many cases an integrable system is equivalent to
	a matrix relation gives the connection with random matrix theory.
	Indeed, when the initial data $\ba(0)$ is chosen randomly, the Lax matrix $L=L(\ba(0))$ becomes 
	a random matrix.
	To define a random initial data, we consider  \emph{invariant measures} with respect to the Hamiltonian flow. In general, such objects have the  form
	\begin{equation}
		\di\mu=m\left( \ba \right)\d a_{1}\wedge\dots\wedge \di a_{n},
		\label{volumeM}
	\end{equation}
	where  the density  $m(\ba)$ is such that the measure $\di\mu\in L^{1}(\mathcal{P})$.
		Hamiltonian systems have a natural invariant measure, the so called  \emph{Gibbs measure} \cite{Khinchin49} obtained from the Hamiltonian itself
	\begin{equation}
		\di\mu_H = \dfrac{1}{Z_H}e^{-\beta H(\ba)}\di\mu,
		\label{eq:gibbs}
	\end{equation}
	with normalization constant $Z_{H}$
	\begin{equation}
		Z_H=\int_{\mathcal{P}}e^{-\beta H(\ba)}\di\mu<\infty
		\label{eq:ZHnorm}
	\end{equation}
	so that $\di\mu_H$ becomes effectively a probability measure.  If the measure is not normalizable in the whole phase
	space $\mathcal P$, one needs to restrict the measure  to a suitable submanifold $\mathcal M$.
	Because of the integrability of the system, we can consider
	a more general probability measure, called generalized Gibbs measure
	\begin{equation}
		\label{Gibbs_generalize} 
		\di\mu_G=\dfrac{1}{Z_G}e^{-\sum_{j=1}^{k}\beta_{j}H_{j}(\ba)}\di\mu,
	\end{equation} 
	where $\beta_1$, \dots, $\beta_k$ are constants and $H_{1}$, \dots, $H_{k}$
	are the first integrals.
	As above, we assume that the normalization constant $Z_{G}$ is finite,
	\begin{equation}
		Z_G=\int_{\mathcal{P}}e^{- \sum_{j=1}^k\beta_jH_j(\ba)}\di\mu<\infty.
		\label{eq:ZGnorm}
	\end{equation}

	In this manuscript, we explore the behaviour of the spectrum of the Lax matrices  of various relevant integrable systems
	when the number of degrees of freedom $N\to\infty$, and
	the initial data is sampled according to a properly chosen Gibbs measure.
	Our main result is that we can compute the density of states for the
	Lax matrix of the exponential Toda lattice and the Volterra lattice.
	This is done through a one-to-one correspondence with the Laguerre
	$\beta$-ensemble at high temperature and with the antisymmetric Gaussian
	$\beta$-ensemble at high temperature, respectively. These are two
	known classes of random matrix ensembles, see \cite{mazzuca2021mean}
	and \cite{Dumitriu_Forrester,Mazzuca2021} respectively.  We consider  other
	relevant cases of integrable systems,  namely the focusing Ablowitz--Ladik lattice \cite{Ablowitz1974,Ablowitz1975}, the focusing Schur flow, and a class of integrable generalization of the Volterra lattice 
	to short range interactions, called the Itoh--Narita--Bogoyavleskii (INB) additive and multiplicative lattices \cite{Bogoyavlensky1991}. In these  cases  the corresponding random  
	Lax matrices are not symmetric nor self-adjoint and 
	we derive numerically their density of states that  has support in the
	complex plane. Interesting patterns of the density of states emerge as we
	vary the parameters of the system.  Finally, for all the 
	integrable systems analysed in this manuscript,  we are able to compute the density of states in the low-temperature limit, namely in the ground state.
	
	This manuscript  is organized as follows. In section \ref{Preliminaries} we give
	an introduction to the basic tools of  the integrable
	system theory  needed in this manuscript.  Next, we study the density
	of states of the random Lax matrices of the exponential Toda lattice,
	the Volterra lattice, the INB lattices, the Ablowitz--Ladik lattice,
	and the Schur flow in sections \ref{toda}, \ref{volterra}, \ref{sec:INB},
	\ref{abllad}, and \ref{schur} respectively.
	Finally, in section \ref{concl} we give some conclusions and an outlook
	for future developments on this topic.

	\section{Preliminaries}
	\label{Preliminaries}

	In this section, we recall some standard tools to study  Hamiltonian integrable systems that we need throughout the manuscript . For further details, we refer to various textbooks and
	monographs \cite{Arnold1997,Arnold1987book,Babelon,Olver1986}.
	
	A \emph{Poisson manifold} is a pair $(\mathcal P,\, \left\{ .\,,. \right\})$
	where $\mathcal P$ is a $n$-dimensional differentiable manifold and
	$\left\{ .\,,. \right\}$ is an antisymmetric bilinear operation on  the space 
	${\mathcal C}^\infty (\mathcal P)$  
	of smooth functions over $\mathcal P$,
	\begin{equation}
		\label{d-pb0}
		\begin{split}
			{\mathcal C}^\infty (\mathcal P)\times &{\mathcal C}^\infty (\mathcal P) \to {\mathcal C}^\infty (\mathcal P)
			\\
			&(f,g) \longrightarrow{} \{ f, g\}
		\end{split}
	\end{equation}
	such that for all functions $f,g,h\in{\mathcal C}^\infty (\mathcal P)$,
	it satisfies:
	\begin{enumerate}
		\item the \emph{Jacobi identity}
		\begin{equation}
			\label{d-jaco}
			\{\{ f, g\}, h\} + \{\{ h, f\}, g\}  + \{\{ g, h\}, f\} =0,
		\end{equation}
		\item the \emph{Leibniz rule}
		\begin{equation}
			\{ hf, g\}=h\{ f, g\}+\{ h, g\}f.
			\label{eq:leibnitz}
		\end{equation}
	\end{enumerate}
	The operator $\left\{ .\,,. \right\}$ is called a \emph{Poisson bracket}.
	When there is no risk of confusion, we simply denote a Poisson
	manifold by $\mathcal P$, where the Poisson bracket is assumed to be fixed
	and given.
	
	In local coordinates $\ba=\left( a_{1},\dots,a_{n} \right)$
	the Poisson bracket is specified by an antisymmetric $(2,0)$  
	tensor $\pi^{ij}(\boldsymbol{a})$, the \emph{Poisson tensor},
	acting on the coordinates as
	\begin{equation}
		\label{Poisson}
		\{a_i,a_j\}=\pi^{ij}(\boldsymbol{a}),\qquad i,j=1,\dots n.
	\end{equation}
	The  Jacobi identity on the coordinates  is equivalent to the  relation
	\begin{equation}
		\frac{\partial \pi^{ij}(\boldsymbol{a})}{\partial a_s }\pi^{sk}(\boldsymbol{a}) 
		+ \frac{\partial \pi^{ki}(\boldsymbol{a})}{\partial a_s }\pi^{sj}(\boldsymbol{a}) 
		+\frac{\partial \pi^{jk}(\boldsymbol{a})}{\partial a_s }\pi^{si}(\boldsymbol{a}) =0, 
		\quad 
		1\leq i<j<k\leq n,
		\label{eq:jacobi}
	\end{equation}
	where we are summing over repeated indices.
	In an open subset of $\mathcal P$ the Poisson tensor has a fixed even rank $2r\leq n$.
	By antisymmetry, it follows that the Poisson tensor can be non-degenerate,
	meaning that $\det \pi(\boldsymbol{a})\neq 0$,
	only if the dimension $n$ of the base space is even, namely $n=2N$.  
	
	Given a function $H(\ba)\in {\mathcal C}^\infty (\mathcal  P)$, it generates
	a set of so-called \emph{Hamilton's equations} through the relation
	\begin{equation}
		\label{H2}
		\dot{a}_{j} =\{a_j,H\}=
		\sum_{j=1}^n\pi^{ij}(\ba)\dfrac{\partial H}{\partial a_j},\quad j=1,\dots, n.
	\end{equation}
	The function $H$ itself is called a \emph{Hamiltonian}.
	The previous set of equations defines a continuum time flow from an initial condition
	$\ba(0)\in \mathcal  P$ to its time evolution  $t>0$, namely 
	$\Phi_{t}\colon \ba(0)\to\ba(t)$.
	A function  $K=K(\boldsymbol{a})$ is constant under evolution 
	$\Phi_{t}$ if and only if 
	\begin{equation}
		\dot{K}=\{K,H\}=0.
		\label{eq:inv}
	\end{equation}
	In this case the quantity $K$ is called a \emph{first integral} or
	a \emph{constant of motion}. 
	The notion of Liouville integrability is strictly
	related to the number of first integrals and the rank of the associated
	Poisson tensor.
	
	\begin{definition}[Liouville integrability]
		A Hamiltonian system \eqref{H2} on a Poisson manifold $ \mathcal P$ of
		rank $2r\leq n$ is \emph{Liouville integrable} if there are $k=n-r$
		first integrals $H_{1}$, \dots, $H_{k}$ in \emph{involution}
		\begin{equation}
			\{H_i,H_j\}=0,\quad i,j=1,\dots, k\,,
			\label{eq:involHiHj}
		\end{equation}
		and functionally independent, namely
		\begin{equation}
			\rank \left( \frac{\partial H_{i}}{\partial a_{j}} \right)_{\substack{i=1,\dots,k \\ j=1,\dots,n}} = k\,,
			\label{eq:funcindip}
		\end{equation}
		\label{def:liouv}
		in a dense subset of $\mathcal P$.
	\end{definition}
	
	As discussed in the Introduction,  random initial data are obtained from an invariant measure $\di\mu$ of the form \eqref{volumeM}. 
	More precisely, this means that the measure of every subset $S\subset  \mathcal P$ with respect to $\di\mu$
	is preserved under the time-evolution $\Phi_t$,
	\begin{equation}
		\int_{\Phi_t(S)} \di\mu=\int_S \di\mu.
		\label{eq:volpreservation}
	\end{equation}
	Interpreting the evolution as a coordinate transformation, we have
	\begin{equation}
		\int_{\Phi_t(S)} \di\mu = \int_S\Phi_t^*(\di\mu),
		\label{eq:volpreservation2}
	\end{equation}
	where $\Phi_t^*(\di\mu)$ is the pull-back of $\di\mu$ through $\Phi_{t}$.
	This shows that the condition \eqref{eq:volpreservation} is satisfied
	if $\Phi_t^*(\di\mu)= \di\mu$.
	In coordinates, for a measure written in the form \eqref{volumeM}, namely $\di\mu=m\left( \ba \right)\d a_{1}\wedge\dots\wedge \d a_{n},$ the invariance of the measure with respect to the Hamiltonian flow
	is expressed by the condition 
	\begin{equation}
		\Div \left( m(\ba) \mathbf{f}_H(\boldsymbol{a})\right) 
		:= \sum_{i=1}^n\dfrac{\partial}{\partial a_i}(m(\ba) (\mathbf{f}_H(\ba))_i)=0,
		\label{eq:divcond}
	\end{equation}
	where $\Div$ is the usual euclidean divergence, see e.g. \cite[Chapter 1]{Guionnet2002}. 
	The vector field $\mathbf{f}_H$ is specified by the Hamiltonian $H$ via the relation
	$(\mathbf{f}_H)_i=\{a_i,H\}$. 
	The condition \eqref{eq:divcond}  can be written in the form
	\begin{equation}
		\{m,H\}+m\Div \left(\mathbf{f}_H\right)=0.
		\label{eq:divm}
	\end{equation}
	
	From formula \eqref{eq:divm} we immediately have two important
	consequences.
	\begin{itemize}
		\item If the Hamiltonian vector field is divergence free,  like in the
		case of a canonical Poisson bracket, it follows  that  the Euclidean measure
		\begin{equation}
			\di\mu_{0} = \d a_{1} \wedge \d a_{2} \wedge \dots \wedge \d a_{n}
			\label{eq:eclmeasure}
		\end{equation}
		is an invariant measure.
		\item If $K$ is a first integral and $m$ is the density of an invariant measure, then
		from the Leibniz rule \eqref{eq:leibnitz} it follows that  
		\begin{equation}
			\label{eq:GGE_formal}
			\wt m := f(K)m
		\end{equation}
		is the density of another invariant measure for every scalar 
		function $f\in\mathcal{C}^{\infty}(\mathcal{P})$.
	\end{itemize}
	In all the examples of this manuscript, the Hamiltonian vector fields
	are divergence free, so we will be allowed to consider the generalized
	Gibbs ensemble described in the Introduction where $\di\mu=\di\mu_0$  the Euclidean measure in  \eqref{eq:eclmeasure}.
	%
	%\begin{remark}
	%The condition \eqref{eq:divcond} depends just on $m(\ba)$ and the vector field $\mathbf{f}_H$, thus  it is independent of the Hamiltonian nature of the dynamical system at hand. In particular, we conclude that a measure $\mu$ as in \eqref{volumeM} is invariant for the dynamical system
	%
	%\begin{equation}
	%    \dot{a}_i = f_i(\ba)\,, \quad i=1,\ldots,N,
	%\end{equation}
	%if and only if \eqref{eq:divcond} holds.
	%\end{remark}

	\section{Laguerre $\beta$-ensemble and the exponential Toda lattice}
	\label{toda}
	
	In this section, we introduce an integrable model that we call \textit{exponential Toda lattice}, since it resembles the well-know Toda lattice. We  construct the Lax pair for this system,  and we define its generalized Gibbs measure. Finally, we compute the  mean density of states  of the Lax matrix.
	
	The exponential Toda lattice is the Hamiltonian system on  $\mathcal P=\mathbb{R}^{2N}$ with canonical Poisson bracket described by the Hamiltonian
	\begin{equation}   \label{HamExpToda}
		H_E({\mathbf p},{\mathbf q}) = \sum_{j=1}^N e^{-p_j} + \sum_{j=1}^{N}e^{q_j-q_{j+1}}\,, \quad p_j,q_j \in \R\,.
	\end{equation}
	We consider periodic boundary conditions
	\begin{equation}
		q_{j+N}=q_{j} + \Delta, \quad p_{j+N} = p_j,\qquad \forall \, j\in \Z,
	\end{equation}
	and $\Delta\geq 0$ is an arbitrary constant. The equations of motion are given in Hamiltonian form as
	\begin{equation} \label{HamiltonianForm}
		\begin{split}
			& \dot{q}_j=\frac{\partial H_E}{\partial p_j} = - e^{-p_j}\,,\\  
			& \dot{p}_j=-\frac{\partial H_E}{\partial q_j} = e^{q_{j-1} -  q_{j}} -  e^{q_{j} -  q_{j+1}}\,.
		\end{split}
	\end{equation}
	The system possesses two trivial constants of motion,
	\begin{equation}    \label{ConstantsOfMotion}
		H_0({\mathbf p},{\mathbf q})= \sum_{j=1}^N (q_j-q_{j+1}),     \qquad    H_1({\mathbf p},{\mathbf q})= \sum_{j=1}^N p_j, 
	\end{equation}
	the first one due to periodicity, the second one due to the translational invariance of the Hamiltonian \eqref{HamExpToda}. 
	In order to obtain a Lax pair for this system we introduce, in the spirit of  Flaschka and Manakov  \cite{Flaschka1974a,Flashka1974b,Manakov1974}, the  variables 
	\begin{equation}    \label{ChangeOfVars}
		x_j = e^{-\frac{p_j}{2}},   \qquad
		y_j = e^{\frac{q_{j}-q_{j+1}}{2}}= e^{-\frac{r_j}{2}},\;\; r_j=q_{j+1}-q_{j},\quad j=1,\dots, N,
	\end{equation}
	where we notice that $\prod_{j=1}^Ny_j=e^{-\frac{\Delta}{2}}$.
	In these variables, the Hamiltonian \eqref{HamExpToda} and the constants of motion \eqref{ConstantsOfMotion}  transform into 
	
	\begin{equation}
		H_E({\mathbf x},{\mathbf y}) = \sum_{j=1}^N (x_j^2+y_j^2),\qquad 
		H_0({\mathbf x},{\mathbf y})= 2 \sum_{j=1}^N \log y_j,     \qquad    H_1({\mathbf x},{\mathbf y})=-2 \sum_{j=1}^N \log x_j.
	\end{equation}
	The Hamilton's equations \eqref{HamiltonianForm} become
	\begin{equation}        \label{ExpToda}
		\dot{x_j} = \frac{x_j}{2} \left( y_j^2-y_{j-1}^2 \right),   \qquad
		\dot{y_j} = \frac{y_j}{2} \left( x_{j+1}^2-x_j^2 \right),\quad j=1,\dots, N,
	\end{equation}
	where  $x_{N+1}=x_1, \, y_0 = y_N$.
	
	One can explicitly construct a Lax pair for this system. Let us introduce  the matrix $E_{r,s}$, defined as $\left(E_{r,s}\right)_{ij}=\delta^i_r \delta^j_s$. Set
	\begin{align}
		\label{eq:L_exp_toda}
		L &= \sum_{j=1}^N (x_j^2+y_{j-1}^2) E_{j,j} + \sum_{j=1}^{N} x_j y_j (E_{j,j+1}+E_{j+1,j}),             \\
		A &= \sum_{j=1}^{N} \frac{x_j y_j}{2} (E_{j,j+1}-E_{j+1,j}),
	\end{align}
	where, accounting for periodic boundary conditions, indices are taken modulo $N$, so that $E_{i,j+N} = E_{i+N,j}= E_{i,j}$ for all $i,j \in \Z$. For example, the matrix $L$ in \eqref{eq:L_exp_toda} has the explicit form
	\begin{equation}
		\label{Lax_exp_toda1}
		L= \begin{pmatrix} x_1^2+y_N^2 & x_1 y_1 & & &x_N y_N \\
			x_1 y_1 & x_2^2+y_1^2 & x_2 y_2 & &\\
			& \ddots & \ddots & \ddots &\\
			&& \ddots & \ddots &  x_{N-1}y_{N-1}  \\
			x_N y_N &&&  x_{N-1}y_{N-1} &x_N^2+y_{N-1}^2\end{pmatrix}\,. 
	\end{equation}
	The system of equations \eqref{ExpToda} then admits the Lax representation 
	\begin{equation}
		\dot{L} = [A,L].
		\label{eq:ccLA}
	\end{equation}
	Hence,  the quantities $H_m = \trace{L^{m-1}}$, $m=2,\ldots,N+1$ are constants of motion as well as the eigenvalues of $L$.  
	For the exponential Toda  lattice, we define the generalized Gibbs ensemble as
	\begin{equation}
		\label{eq:Guibbs_exp_toda}
		\di\mu_{ET} = \frac{1}{Z^{E}_N(\beta,\eta,\theta)} \exp\left( -\beta H_E +  \theta H_0 -\eta H_1  \right)\di \br\di\bp\,,
	\end{equation}
	where  $\beta,\eta,\theta >0$, the Hamiltonians $H_E$, $H_0$ and $H_1$ are defined in \eqref{HamExpToda} and \eqref{ConstantsOfMotion} respectively, $Z_N^{E}$ is the normalization constant,  $\di\br = \di r_1\ldots\di r_N$ and analogously for $\di \bp$. We notice that according to this measure, all the variables are independent, moreover all $p_j$ are identically distributed, and so are the $r_j$. After introducing the variables $(\br,\bp)\to(\bx,\by)$,   the previous measure turns into 
	\begin{equation}
		\label{Gibbs_measure_TodaExp}
		\di\mu_{ET} = \frac{1}{Z^{H_E}_N(\beta,\eta,\theta)} \prod_{j=1}^N x_j^{2\eta -1}e^{-\beta x_j^2}\di x_j \prod_{j=1}^N  y_j^{2\theta -1}e^{-\beta y_j^2 }  \di y_j\,.
	\end{equation}
	Let $\chi_{2\alpha}$ be the \emph{chi-distribution}, defined by its density
	\begin{equation}
		f_{{2\alpha}}(r)= \frac{r^{2\alpha-1}e^{-\frac{r^2}{2}}}{2^{\alpha-1}\Gamma(\alpha)}, \qquad r\in\R^+,
	\end{equation}
	here $\Gamma$ is the classical gamma function \cite[§5]{dlmf}.  Then,  the variables $x_j$ and $y_j$ in the Gibbs measure \eqref{Gibbs_measure_TodaExp} are independent random variables with scaled $chi$-distribution, respectively $f_{2\eta}(\sqrt{2\beta}x_j)\sqrt{2\beta}\di x_j$ and $f_{2\theta}(\sqrt{2\beta}y_j)\sqrt{2\beta} \di y_j$.
	
	The Lax matrix $L$ in \eqref{Lax_exp_toda1} becomes a random matrix when the entries are sampled according to \eqref{Gibbs_measure_TodaExp}. Such random matrix can be linked to 
	the so-called \emph{Laguerre $\alpha$-ensemble} \cite{mazzuca2021mean}. The connection is obtained noticing that the matrix $L$  admits the following decomposition
	\begin{equation}
		\label{Cholesky}
		L = B B^\intercal, \qquad B = \sum_{j=1}^N x_j E_{j,j} + \sum_{j=1}^{N} y_j E_{j+1,j},
	\end{equation}
	where $B^\intercal$ is the matrix transpose.
	On the other hand, the Laguerre $\alpha$-ensemble  is given by the set of matrices  
	%$L_{\alpha,\gamma} = B_{\alpha,\gamma} (B_{\alpha,\gamma})^\intercal$, where $B_{\alpha,\gamma} \in \text{Mat}(N\times M),\, M \geq N$, and
	\begin{equation}
		\label{Laguerre_alpha}
		L_{\alpha,\gamma} = B_{\alpha,\gamma} (B_{\alpha,\gamma})^\intercal,\quad     B_{\alpha,\gamma} =\frac{1}{\sqrt{2}} \begin{pmatrix}
			\begin{matrix}
				x_1 \\
				y_1 & x_2 \\
				& \ddots & \ddots  \\
				&& y_{N-1} & x_N  
			\end{matrix}& \bigzero{N}{(M-N)}
		\end{pmatrix},
	\end{equation}
	here $\bigzero{N}{(M-N)}$ is the zero matrix of dimension $N\times (M-N)$. The variables $x_n,y_n$ are distributed according to  chi-distribution
	\begin{align}
		\label{eq:LaguerreBmat}
		x_n &\sim \chi_{\frac{2\alpha}{\gamma}}\quad n = 1, \ldots, N,\\
		y_n &\sim \chi_{2\alpha}\quad n= 1, \ldots, N-1.
	\end{align}
	Thus, the following entry wise measure on the matrices $B_{\alpha,\gamma}$ can be defined,
	\begin{equation}    \label{GibbsLaguerre}
		\di\mu_{B_{\alpha,\gamma}} = \frac{1}{\left(2^{\frac{\alpha}{\gamma}-1} \Gamma(\frac{\alpha}{\gamma})\right)^{N}\left(2^{\alpha-1} \Gamma(\alpha)\right)^{N-1}} \prod_{j=1}^N x_j^{\frac{2\alpha}{\gamma} -1} e^{-\frac{x_j^2}{2}}\di x_j \prod_{j=1}^{N-1} y_j^{2\alpha -1} e^{-\frac{y_j^2}{2}}\di y_j\,.
	\end{equation}

	We observe that the matrix $B$ in \eqref{Cholesky}  has the same form of $B_{\alpha,\gamma}$ in \eqref{eq:LaguerreBmat}, with the addition of the corner element $y_N E_{1,N}$.  Furthermore, the  rescaling  of the variables  $(x_j,y_j)\mapsto\frac{1}{\sqrt{2\beta}}(x_j,y_j)$ in \eqref{Gibbs_measure_TodaExp},  amounts to the matrix rescaling $B\mapsto \frac{1}{\sqrt{2\beta}}B$, and comparing with \eqref{GibbsLaguerre} we see that $\frac{1}{\sqrt{2\beta}}B$ is a rank one perturbation of the matrix $B_{\theta,\frac{\theta}{\eta}}$.

	We are interested in studying the density of states $ \nu_{ET}$ for the Lax matrix $L$ when the entries are distributed according to the Gibbs measure $\di\mu_{ET}$   in \eqref{Gibbs_measure_TodaExp}. The density of states  $\nu_{ET}$  is  obtained from the weak convergence of the empirical measure of the Lax matrix $L$, namely
	% defined as the non-random probability measure
	\begin{equation} 
		\label{eq:dos}
		\frac{1}{N}\sum_{j=1}^N \delta_{\lambda_j}  \xrightharpoonup{N\to\infty} \nu_{ET}\,,
	\end{equation}
	%where with $\rightharpoonup$ we denote the weak convergence in $L^{1}(\R)$, 
	where $\lambda_j$ are the eigenvalues of $L$ and $\delta_x$ is the Dirac delta function centred at $x$. 
	
	In order to study the density of states of the exponential Toda lattice, we need the following result proved in  \cite{mazzuca2021mean}.
	
	\begin{theorem}[cf. \cite{mazzuca2021mean}, Theorem 1.1]
		\label{thm:mazzuca}
		Consider the matrix $L_{\alpha,\gamma} = B_{\alpha,\gamma}B_{\alpha,\gamma}^\intercal$ distributed according to $\di\mu_{B_{\alpha,\gamma}}$  in \eqref{GibbsLaguerre}. Then, its mean density of states $\nu_{L_{\alpha,\gamma}}$ takes the form 
		\begin{equation}
			\nu_{L_{\alpha,\gamma}} =\partial_\alpha\left(\alpha  \mu_{\alpha,\gamma}(x)\right) \di x\, ,\quad x \geq 0 ,
		\end{equation}
		where
		\begin{equation}
			\label{mualpha}
			\mu_{\alpha,\gamma}(x) := \frac{1}{\Gamma(\alpha+1)\Gamma\left(1+ \frac{\alpha}{\gamma} + \alpha\right)} \frac{x^\frac{\alpha}{\gamma}e^{-x}}{\Big\lvert \psi\left(\alpha,-\frac{\alpha}{\gamma};xe^{-i\pi}\right)\Big\rvert^2}, \qquad x \geq 0\,,
		\end{equation}
		 and here $\psi(v,w;z)$ is  the Tricomi's confluent hypergeometric function (see Appendix \ref{appendix}).
	\end{theorem}
	\begin{remark} The proof of theorem~\ref{thm:mazzuca} has been obtained by comparing the Laguerre  $\alpha$-ensemble \eqref{Laguerre_alpha} with the Laguerre
		$\beta$-ensemble at high temperature \cite{Allez2013}.
	\end{remark}
	The following  corollary follows.
	\begin{corollary} \label{thm:ExpToda Gibbs}
		Consider the Lax matrix $L=BB^\intercal$ in \eqref{Cholesky} of the exponential Toda lattice with Hamiltonian \eqref{HamExpToda} and   endow the entries of the matrix $B$ in \eqref{Cholesky} with the Gibbs measure $\mu_{ET}$ \eqref{Gibbs_measure_TodaExp}. Then, the density of states $\nu_{ET}$ of the Lax matrix $L=BB^\intercal$ takes the form
		\begin{equation}
			\label{nu_E}
			\nu_{ET} = \beta\partial_{\alpha}(\alpha \mu_{\alpha,\gamma}(\beta x))_{\vert_{\alpha = \theta, \gamma =\frac{\theta}{\eta} }} \di x, \quad x\geq 0\,,
		\end{equation}
		where the density $ \mu_{\alpha,\gamma}$ is defined in \eqref{mualpha}.
		%\begin{equation}
		%     \mu_{\alpha,\gamma}(x) := \frac{1}{\Gamma(\alpha+1)\Gamma\left(1+ \frac{\alpha}{\gamma} + \alpha\right)} \frac{x^\frac{\alpha}{\gamma}e^{-x}}{\Big\lvert \psi\left(\alpha,-\frac{\alpha}{\gamma};xe^{-i\pi}\right)\Big\rvert^2}, \qquad x \geq 0\,,
		%\end{equation}
		%here $\psi(v,w;z)$ is  the Tricomi's confluent hypergeometric function \cite[§13]{dlmf}.
	\end{corollary}
	\begin{proof}
		
		First, we notice that by virtue of general theory, see \cite[Theorem A.43]{BaiBook}, we can restrict  to the case $y_N = 0$ in \eqref{Gibbs_measure_TodaExp}. As observed above, performing the change of variables $(x_j,y_j)\mapsto\frac{1}{\sqrt{2\beta}}(x_j,y_j)$, which amounts to rescale $B\mapsto \frac{1}{\sqrt{2\beta}}B$,  one has that the matrix entries of $\frac{1}{\sqrt{2\beta}}B$  are distributed as the matrix entries of $B_{\theta,\frac{\theta}{\eta}}$. 
		Applying  Theorem \ref{thm:mazzuca} we obtain the claim.
	\end{proof}

	\subsection{Parameter Limit}
	\label{sec:ExpTodaLimit}
	
	In this section, we examine the low-temperature limit of the Hamiltonian system \eqref{HamExpToda}. Namely, we want to compute the eigenvalues of the  Lax matrix $L$ in \eqref{eq:L_exp_toda} in the limit $\beta,\theta,\eta \to \infty$,  in such a way that 
	$$\eta = \wt \eta \beta, \quad \theta = \wt \theta \beta ,$$  where $\wt \eta$ and $\wt\theta$ are in compact sets of $ \R_+$. 
	
	Since  all $x_j$ and $y_j$ are independent random variables, we just have to consider the weak limit of the rescaled chi-distributions, respectively
	\[
	f_{2\tilde{\eta}\beta}(\sqrt{2\beta}x)\sqrt{2\beta}\di x,\quad f_{2\tilde{\theta}\beta}(\sqrt{2\beta}y)\sqrt{2\beta} \di y \,.
	\]
	We explicitly work out one of the cases above.

	We consider a continuous and bounded  function  $h\;:\; \R_+ \to \R$  and  evaluate  the limit 
	\begin{equation}
		\lim_{\beta\to\infty}\int_0^\infty h(x) f_{2\eta\beta}(\sqrt{2\beta}x)\sqrt{2\beta}\di x = \lim_{\beta\to\infty} \frac{\int_0^\infty h(x) e^{\beta(2\wt\eta\log x- x^2)} \di x}{\int_0^\infty x^{2\wt\eta\beta} e^{-\beta x^2} \di x}=h\left(\sqrt{\wt\eta}\right)\,.
	\end{equation}
	The last identity has been obtained by applying the Laplace method (see \cite{MillerBook}) and observing that the  minimizer  of the  term $2\wt \eta \log(x) - x^2$  in the  exponent
	of the integral  is $x_0 = \sqrt{\wt \eta}$.

	As a consequence, we conclude that $    x_j \rightharpoonup \sqrt{\wt\eta}$ and $ y_j \rightharpoonup \sqrt{\wt \theta}$, $j=1,\ldots,N$  as $\beta\to\infty$, 
	where with $\rightharpoonup$ we denote the weak convergence. The previous limit implies that the measure $\nu_{ET}$ in \eqref{nu_E}  weakly converges, in the low temperature limit, to the density of states of the matrix $L_\infty$
	\begin{equation}
		L_\infty = \begin{pmatrix} \wt \eta + \wt \theta & \sqrt{\wt\eta\wt\theta} & & &\sqrt{\wt\eta\wt\theta} \\
			\sqrt{\wt\eta\wt\theta} & \wt \eta + \wt \theta & \sqrt{\wt\eta\wt\theta} & &\\
			& \ddots & \ddots & \ddots &\\
			&& \ddots & \ddots & \sqrt{\wt\eta\wt\theta}  \\
			\sqrt{\wt\eta\wt\theta} &&&\sqrt{\wt\eta\wt\theta}&\wt \eta + \wt \theta\end{pmatrix}\,. 
	\end{equation}

Indeed, the fact that $L$ is tridiagonal with iid entries along the diagonals, implies its $k$-th moment depends on a multiple of $k$ number of variables only; specifically looking back at the Lax matrix $L$ in \eqref{Lax_exp_toda1}
\begin{equation}
	\la \tr{L^k }\ra  =\la \sum_{j=1}^N \left( L^k \right)_{jj}  \ra = N \cdot \left \langle \left( L^k \right)_{11} \right \rangle =: N \cdot \left \langle f \left(x_{N-k},\dots, x_{N+k}; y_{N-k},\dots, y_{N+k} \right) \right \rangle,
\end{equation}
for some function $f(\cdot)$ of its entries. Then, passing to the density of states \eqref{eq:dos} and renaming the iid variables, the scaling factor $N$ identically cancels out and moments converge,
\begin{equation}
	\lim_{N \to \infty} \frac{1}{N} \left \langle \tr{L^k}\right \rangle = \left \langle f \left(x_{1},\dots, x_{2k}; y_{1},\dots, y_{2k} \right) \right \rangle.
\end{equation}
The eigenvalues being functions of moments, density of states converges as well. In particular, this also shows that the two limit commute in taking the density of states at low temperature,
\begin{equation} 
	\lim_{N \to \infty} \lim_{\beta \to \infty} \frac{1}{N}\sum_{j=1}^N \delta_{\lambda_j} =   \lim_{\beta\to \infty} \lim_{N \to \infty} \frac{1}{N}\sum_{j=1}^N \delta_{\lambda_j} = \nu_{\infty}\,,
\end{equation}
since the limits can be passed directly to the variables $x_i,y_i$.

	The matrix $L_\infty$ is a circulant matrix, so its eigenvalues can be computed explicitly \cite{gray2006toeplitz}  as 
	\begin{equation}
		\lambda_j = \wt \eta + \wt \theta + 2\sqrt{\wt \eta\wt\theta }\cos\left(2\pi \frac{j}{N}\right)\,, \quad j=1,\ldots,N\,.
	\end{equation}
	From this explicit expression, it follows that the density of states of $L_\infty$ is 
	\begin{equation}
		\label{nu_L}
		\nu_{L_\infty} = \frac{1}{2\pi}   \frac{dx}{\sqrt{(c_+-x)(x-c_-)}} \mathbbm{1}_{(c_-,c_+)},\quad c_{\pm}=(\sqrt{\wt\eta}\pm\sqrt{\wt\theta})^2. 
		%    \frac{1}{\sqrt{4\wt\eta\wt\theta - (x-\wt\eta -\wt \theta)^2}} \mathbbm{1}_{(\wt \eta + \wt \theta - 2\sqrt{\wt\eta\wt\theta}, \wt \eta + \wt \theta - 2\sqrt{\wt\eta\wt\theta})}(x) \di x\,,
	\end{equation}
	% The density of states of 
	%$ \nu_L$ in the low-temperature limit, is equal to
	%
	%
	here $\mathbbm{1}_{(c_-,c_+)}$ is the indicator function of the set $(c_-,c_+)$. In particular, this measure belongs to the class of Arcsin distributions.
	Thus, we proved the following result.
	
	\begin{proposition}
		Consider the random Lax matrix $L$ in  \eqref{eq:L_exp_toda}  sampled from  the Gibbs ensemble $\di\mu_{ET}$ \eqref{Gibbs_measure_TodaExp} of the exponential Toda lattice \eqref{ExpToda}. The density of states of the matrix $L$ in the low-temperature limit, i.e. when $\beta,\theta,\eta\to\infty$ in such a way that $\eta = \wt \eta\beta,\, \theta = \wt \theta \beta$, with  $\wt \eta,\wt \theta$ in compact subsets of   $ \R_+$, is the Arcsin distribution  given by \eqref{nu_L}.

	\end{proposition}

	\section{Volterra lattice}
	\label{volterra}
	
	The \emph{Volterra lattice}, also known as the \emph{discrete KdV equation}, describes the  motion of $N$ particles on the line with equations 
	\begin{equation}    \label{Volterra}
		\dot{a_j} = a_j \left(a_{j+1} - a_{j-1} \right), \qquad j=1,\dots,N.
	\end{equation}
	It was originally  introduced  by Volterra to study 
	population evolution in a hierarchical system of competing species. It was first
	solved by Kac and van Moerbeke in \cite{Kac1975} using a discrete version of inverse scattering due to Flaschka \cite{Flashka1974b}. Equations \eqref{Volterra} can be considered as a finite-dimensional approximation of the Korteweg–de Vries equation. 
	
	The phase space is $\mathbb{R}_+^N$ and we consider  periodic boundary conditions $a_j=a_{j+N}$ for all $j\in\Z$. The Volterra lattice  is a  reduction of the \emph{second flow of the Toda lattice} \cite{Kac1975}. Indeed, the latter is described by the dynamical system
	\begin{align}
		\dot{a_j} &= a_j \left(b_{j+1}^2 - b_{j}^2 + a_{j+1}-a_{j-1} \right),     &   &\qquad j=1,\dots,N,    \\
		\dot{b_j} &= a_j(b_{j+1}+b_j) -a_{j-1}(b_j+b_{j-1}),       &   &\qquad j=1,\dots,N,     \label{2ndToda}
	\end{align}
	and equations \eqref{Volterra} are recovered just by setting $b_j\equiv0$. The Hamiltonian structure of the equations follows from the one of the Toda lattice.  On the phase space $\R^N_+$  we introduce the   Poisson bracket
	\begin{equation}
		\label{eq:poisson_volterra}
		\{ a_j, a_i \}_{\Volt} =  a_ja_i(\delta_{i,j+1} - \delta_{i,j-1})\,
	\end{equation}
	and the Hamiltonian  $H_1 = \sum_{j=1}^N a_j\,$ so that the equations of motion \eqref{Volterra}  can be written in the Hamiltonian form 
	\begin{equation}
		\dot{a}_j = \{ a_j, H_1\}_{\Volt}\,.
		\label{eq:hamvoltN}
	\end{equation}
	An elementary constant of motion for the system is $H_0 = \prod_{j=1}^N a_j$ that is  independent of  $H_1$.
	
	The Volterra lattice  is a completely integrable system, and  it admits several equivalent \emph{Lax representations}, see e.g. \cite{Kac1975,Moser75}. The classical one reads
	\begin{equation}  
		\dot{L}_1 = \left[A_1,L_1\right],
	\end{equation}
	where 
	\begin{equation}
		\begin{split}
			\label{eq:classic_Vlax}
			L_1 &= \sum_{j=1}^{N} a_{j+1}E_{j+1,j}+E_{j,j+1},             \\
			A_1 &= \sum_{j=1}^{N} (a_j + a_{j+1}) E_{j,j} + E_{j,j+2}\,,     
		\end{split}
	\end{equation}
	where we recall that the matrix $E_{r,s}$ is defined as $\left(E_{r,s}\right)_{ij}=\delta^i_r \delta^j_s$ and   $E_{j+N,i} =E_{j,i+N} = E_{j,i}$.
	There exists also a \emph{symmetric} formulation due to Moser \cite{Moser75},
	\begin{equation}
		\label{Vlax}
		\begin{split}
			\dot{L}_2 & = \left[A_2,L_2\right]\\
			L_2 &= \sum_{j=1}^{N} \sqrt{a}_j (E_{j,j+1}+E_{j+1,j})\,,             \\
			A_2 &= \frac{1}{2}\sum_{j=1}^{N} \sqrt{a_ja_{j+1}} (E_{j,j+2}-E_{j+2,j})\,,     
		\end{split}
	\end{equation}
	which assumes that all $a_j > 0$.

	Furthermore, we point out that there exists also an \textit{antisymmetric} formulation for this Lax pair, indeed a straightforward computation yields 
	\begin{proposition} \label{AntiProp}
		Let $a_j > 0$ for all $j=1,\ldots,N$. Then, the dynamical system \eqref{Volterra} admits an antisymmetric Lax matrix $L_3$ with companion matrix $A_3$, namely the equations of motion are equivalent to $\dot{L_3} = \left[A_3,L_3\right]$ with
		\begin{align}
			\label{LaxVolterra}
			L_3 &= \sum_{j=1}^{N}\sqrt{ a_j} (E_{j,j+1}-E_{j+1,j}),             \\
			A_3 &= \frac{1}{2}\sum_{j=1}^{N} \sqrt{a_ja_{j+1}} (E_{j+2,j}-E_{j,j+2}).      \label{AntiVlax}
		\end{align}
	\end{proposition}

	\subsection{Gibbs Ensemble}
	
	We  introduce a Gibbs ensemble for the Volterra lattice \eqref{Volterra} by observing 
	that   its  vector field $f_j=a_j \left(a_{j+1} - a_{j-1} \right)$ is divergence free, due to the periodic boundary conditions. Therefore, an invariant measure can be obtained from 
	\eqref{eq:GGE_formal}. We use $ H_0 = \prod_{j=1}^N a_j,$ and $H_1 = \sum_{j=1}^N a_j$ as constants of motion to construct the invariant measure
	\begin{equation}
		\label{VolterraGibbs}
		\di\mu_{\Volt}(\ba) = \frac{1}{Z_N^{\Volt}(\beta,\eta)}e^{-\beta H_1 + (\eta-1)\log H_0} \d  \ba,\quad \beta,\eta>0,
	\end{equation}
	where 
	\begin{equation}
		Z_N^{\Volt}(\beta,\eta) = \left( \frac{\Gamma\left( \eta \right)}{\beta^{\eta}} \right)^{N}<\infty,
		\label{eq:ZNvolterravalue}
	\end{equation}
	and $\Gamma\left( \eta \right)$ is the Euler gamma function \cite[§5]{dlmf}. We notice that according to this measure, all the variables are independent and identically distributed (i.i.d.).
	
	Next we want to characterize the density of states of the antisymmetric Lax  $L_3$  of the Volterra lattice  given in Proposition \ref{AntiProp}.
	Among the three Lax matrices of the Volterra lattice, the matrix $L_3$ is particularly useful since it allows us  to connect the Volterra lattice with a specific $\alpha$-ensemble, namely the antisymmetric Gaussian $\alpha$-ensemble. 
	The antisymmetric Gaussian $\alpha$-ensemble, see \cite{Mazzuca2021}, is the family of random
	antisymmetric tridiagonal matrices 
	\begin{equation}
		\label{eq:alphaense}
		L_\alpha = 	\begin{pmatrix}
			0 & y_1 \\
			-y_1 & 0 & y_2 \\
			&\ddots &\ddots &\ddots\\
			&& -y_{N-2} & 0 & y_{N-1}\\
			&&& -y_{N-1} & 0 
		\end{pmatrix},
	\end{equation}
	where $y_i$ are i.i.d. random variables  with density
	\begin{equation}
		f_{2\alpha}(y)=\frac{y^{2\alpha-1}e^{-y^2}}{\Gamma(\alpha)}\,,\quad y\in\R^+\,,
	\end{equation}
	which is just a rescaled chi-distribution. Even though we use a different expression of the  chi-distribution  with respect to section~\ref{toda}, we keep the same notation $ f_{2\alpha}(y)$ for the density that will be used only in this section.
	This distribution induces a measure on the entries of the matrix $L_{\alpha}$, namely
	\begin{equation}    \label{Alfa Gauss Measure}
		\di\mu_{L_{\alpha}} = \frac{\prod_{i=1}^{N-1} y_i^{2\alpha-1}e^{-y_i^2}\,\mathbbm{1}_{\R_+}(y_i)\di\by}{\Gamma(\alpha)^{N-1}}.
	\end{equation}
	In \cite{Mazzuca2021} the authors studied this matrix ensemble in connection with the antisymmetric Gaussian $\beta$-ensemble introduced by Dumitriu and Forrester \cite{Dumitriu_Forrester}
	in the  high temperature regime, and computed explicitly its density of states $ \nu_{L_{\alpha}}(x)$, defined as
	\begin{equation}
		\dfrac{1}{N}\sum_{j=1}^{N} \delta_{\Im(\lambda_j)} \xrightharpoonup{N\to\infty}  \nu_{L_{\alpha}}(x)\,, 
	\end{equation}
	where $\lambda_j$ are the eigenvalues of $L_\alpha$. Since the matrix $L_\alpha$ is antisymmetric with real entries, its eigenvalues are purely imaginary numbers. 
	
	\begin{theorem}[cf. \cite{Mazzuca2021}]
		
		\label{thm:guido_peter}
		The  density of states  of  the random matrix $L_\alpha$ in \eqref{eq:alphaense}, is explicitly given by
		\begin{equation}
			\nu_{L_\alpha}(x) = \partial_\alpha(\alpha \theta_\alpha(x)) \di x\,,
		\end{equation}
		where
		\begin{equation}
			\label{eq:theta}
			\theta_\alpha(x) = \left \vert \Gamma(\alpha) W_{-\alpha + 1/2,0}(-x)\right\vert^{-2}\,,
		\end{equation}
		here $\Gamma(x)$ is the gamma function and $W_{k,\mu}(z)$ is the Whittaker function (see Appendix \ref{appendix}) .
	\end{theorem}
	
	\begin{remark} The proof of Theorem~\ref{thm:guido_peter} has been obtained in \cite{Mazzuca2021} by comparing the antisymmetric Gaussian   $\alpha$-ensemble \eqref{eq:alphaense} with the antisymmetric Gaussian $\beta$-ensemble at high temperature, which was considered in the same paper.
	\end{remark}
	
	We notice that performing the change of coordinates $a_j = x_j^2$, the Gibbs ensemble \eqref{VolterraGibbs} reads:
	\begin{equation}
		\di\mu_{\Volt}(\bx) = \frac{\prod_{j=1}^N x_j^{2\eta-1}e^{-\beta \sum_{j=1}^N x^2_j}\mathbbm{1}_{\R_+}(x_j)\di \bx }{Z_N^{\Volt}(\beta,\eta)}\,,
	\end{equation}
	which, up to a rescaling $x_j \to x_j/\sqrt{\beta}$ and for the extra term $x_N$ in the  probability distribution, is exactly the distribution \eqref{Alfa Gauss Measure} of the matrix $L_\alpha$. Furthermore, the matrix $L_3$ is a 2 rank perturbation of the matrix $L_\alpha$.  Therefore, by a corollary of \cite[Theorem A.41]{BaiBook} and Theorem \ref{thm:guido_peter}, we obtain the following.
	
	\begin{corollary}
		Consider the matrix $L_3$ in \eqref{AntiVlax}   endowed  with the Gibbs measure $\di\mu_{\Volt}$ \eqref{VolterraGibbs}. Then, the density of states of the matrix $ L_3$ is explicitly given by
		\begin{equation}
			\nu_{\Volt}(x) = \sqrt{\beta}\partial_\eta\left( \eta\theta_{\eta}(\sqrt{\beta}x)\right)\di x\,,
		\end{equation}
		where $\theta_\alpha(x)$ is given in \eqref{eq:theta}.
	\end{corollary}

	\subsection{Parameter Limit}
	\label{SEC:PAR_LIM_VOLT}
	
	As for the case of exponential Toda \eqref{ExpToda}, in this section we consider the low-temperature regime of the Volterra lattice, namely the limit $\eta,\beta\to\infty$, in such a way that $\eta = \beta \wt\eta$,  with  $\wt \eta$  in a compact set of $ \R_+$, and we compute the density of states of the matrix $ L_3$ \eqref{LaxVolterra} in this regime.

	Applying the same techniques of Section \ref{sec:ExpTodaLimit}, we conclude that the density of states of the matrix $ L_3$  in the low-temperature limit coincides with the one of the matrix $L_\infty$, where

	\begin{equation}
		L_\infty = \begin{pmatrix} 0 & \sqrt{\wt \eta}&&& -\sqrt{\wt \eta} \\
			-\sqrt{\wt \eta}& 0 & \sqrt{\wt \eta} &&& \\
			&\ddots & \ddots & \ddots && \\
			&&\ddots&\ddots&\sqrt{\wt\eta}\\
			\sqrt{\wt \eta}&&&-\sqrt{\wt \eta}&0\end{pmatrix}\,.
	\end{equation}
	Since the matrix $L_\infty$ is circulant, we can readily compute its eigenvalues as
	\begin{equation}
		\lambda_j = 2i\sqrt{\wt\eta}\sin\left(2\pi \frac{j}{N} \right)\,, \quad j = 1,\ldots, N\,.
	\end{equation}
	From this explicit formula, it follows that the density of states of the matrix $ L_\infty$ reads 
	\begin{equation}
		\nu_{ L_\infty} = \frac{1}{2\pi} \frac{1}{\sqrt{4\wt\eta - x^2}} \mathbbm{1}_{(-2\sqrt{\wt\eta},2\sqrt{\wt\eta})}(x)\di x\,.
	\end{equation}
	Such measure coincides with  the measure  $ \nu_{L^\alpha}$ in the low-temperature limit, and it belongs to the class of Arcsin distributions. Thus, we just proved the following.
	
	\begin{proposition}
		Consider the Gibbs ensemble $\di\mu_{Volt}$ of the Volterra lattice \eqref{VolterraGibbs}, in the low-temperature limit, i.e. $\beta,\eta \to\infty$, in such a way that $\eta = \wt \eta\beta$, where
		$\wt \eta$ is in a compact subset of  $\R_+$. Then, the density of states $\nu_{\Volt}$ of the Lax matrix $L_3$ in  \eqref{LaxVolterra}  converges,  in this regime,  to
		
		\begin{equation}
			\nu_{\Volt} = \frac{1}{2\pi} \frac{1}{\sqrt{4\wt\eta - x^2}} \mathbbm{1}_{(-2\sqrt{\wt\eta},2\sqrt{\wt\eta})}(x)\di x\,.
		\end{equation}
	\end{proposition}
	
	\section{Generalization of the Volterra lattice: the INB $k$-lattices}
	\label{sec:INB}
	The Volterra lattice \eqref{Volterra} can be generalized  in a variety of ways. The most natural ones are two families of lattices described in
	\cite{Bogoyavlensky1991}  (see also \cite{Bogoyavlensky1988,Itoh1975,Narita1982} ) which include short range interactions.
	The first  family is called  \emph{additive
		Itoh--Narita--Bogoyavleskii (INB) $k$-lattice} and is defined by the equations
	\begin{align}
		\dot{a}_{i} &= a_{i} \left( \sum_{j=1}^{k} a_{i+j}- \sum_{j=1}^{k} a_{i-j} \right),
		\quad
		i = 1,\dots,N,\;N\geq k\in N.
		\label{eq:inbpadd}
		\end{align}
		The second family is called the
	\emph{multiplicative}  \emph{ Itoh--Narita--Bogoyavleskii (INB) $k$-lattice} and is defined by the equations
		\begin{align}
		\dot{a}_{i} &= a_{i} \left( \prod_{j=1}^{k} a_{i+j}- \prod_{j=1}^{k} a_{i-j} \right),
		\quad
		i = 1,\dots,N,\;N\geq k\in N.
		\label{eq:inbpmul}
	\end{align}
	In both cases we consider  the periodicity condition
	$a_{j+N}=a_{j}$ holds.\\
	Setting $k=1$, we recover from both lattices the Volterra one
	\eqref{Volterra}.  Further generalizations of the INB lattice were recently considered in
	\cite{EKV}.
	
	%Notice that the flow corresponding to the $k^{th}$
	%element of either family does not commute with the other flows of the
	%family; rather, it serves as the simplest member of an integrable hierarchy
	%on its own. That is, the families of INB equations \eqref{eq:inbpadd}
	%and \eqref{eq:inbpmul} are infinite families of integrable systems
	%indexed by $k\in\N$.
	A crucial difference in the two models is that in the additive lattice \eqref{eq:inbpadd}
	the interaction is on arbitrary number of points,  but  the non-linearity is still
	quadratic like the original Volterra lattice \eqref{Volterra}; on the other hand
	, the  multiplicative lattice \eqref{eq:inbpmul} 
	admits non-linearity of \emph{arbitrary order}.
	Moreover, both families admit the KdV equation as continuum limits,
	see \cite{Bogoyavlensky1991}.

	As mentioned earlier, the additive INB $k$-lattice is an integrable system for
	all $k\in\N$ and $i\in\Z$, since they all admit a Lax pair formulation \eqref{Lax0}.
	For the additive INB lattice \eqref{eq:inbpadd}, it reads
	\begin{align}
		L^{(+,k)} & = \sum_{i=1}^{N}
		\left( a_{i+k} E_{i+k,i} + E_{i,i+1}\right),
		\label{eq:inbpLadd}
		\\
		A^{(+,k)} & = 
		\sum_{i=1}^N\left(\sum_{j=0}^{k} a_{i+j}\right) E_{i,i}
		+ E_{i,i+k+1} \, ,
		%\right],
		\label{eq:inbpMadd}
	\end{align}
	we recall that we are always considering periodic boundary conditions, 
	so for all $j \in\Z$, $a_{j+N} = a_j$ and $E_{i,j+N} =E_{i+N,j} = E_{i,j}$.
	The constants of motion obtained through this Lax pair are in involution with respect
	to the Poisson bracket
	\begin{equation} 
		\left\{ a_{j},a_{i} \right\}_{(+,k)} = 
		a_{j}a_{i} \left( \sum_{s=1}^{k} \delta_{j+s,i} - \sum_{s=1}^{k}\delta_{j-s,i} \right).
		\label{eq:pbinbadd}
	\end{equation}
	Then, the additive INB $k$-lattice \eqref{eq:inbpadd} can be written as
	\begin{equation}
		\dot{a}_{i} = \left\{ a_{i}, H_{1} \right\}_{(+,k)},
		\label{eq:inbadd}
	\end{equation}
	where the Hamiltonian function $H_{1}=\sum_{j=1}^{N} a_{j}$ is the same
	as in equation \eqref{eq:hamvoltN}. In the same way, it is possible to
	prove that the function $H_{0}=\prod_{j=1}^{N}a_{j}$ is a first
	integral for the additive INB $k$-lattice \eqref{eq:inbpadd} as well.

	Similarly, the multiplicative INB $k$-lattices can be endowed with a Lax Pair for all $k\in\N$, therefore it is another example of integrable systems.
	Specifically, for the periodic case we presented in equation~\eqref{eq:inbpmul}, the Lax pair
	reads
	\begin{align}
		\label{eq:inbpLmul}
		L^{(\times,k)} & = \sum_{i=1}^{N}\left( a_{i} E_{i,i+1} 
		+E_{i+k,i} \right),\\
		\label{eq:inbpMmul}
		A^{(\times,k)} & = \sum_{i=1}^{N} \left(\prod_{j=0}^{k} a_{i+j}\right) E_{i,i+k+1}\,.
	\end{align}
	
	We notice that both $H_1 = \sum_{j=1}^Na_j$, and $H_0 = \prod_{j=1}^N a_j$ are constants of motion for these systems, for all $k \in \mathbb{N}$.
	
	\begin{remark}
		For fixed $k$, there exists a transformation that maps the multiplicative INB $k$-lattice to the additive one. Namely, consider the system \eqref{eq:inbpmul} and define the new set of variables
		\begin{equation}
			b_i := a_i \cdot \dots \cdot a_{i+k-1},
			\quad
			i = 1,\dots,N,
			\label{eq:MultToAdd}
		\end{equation}
		where the indices are taken modulo $N$. Then, it is immediate to see that
		\begin{equation}
			\dot{a}_{i} = a_{i} \left( b_{i+1} - b_{i+k-1} \right),
			\quad
			i = 1,\dots,N,
		\end{equation}
		which in turn, due to telescopic summations, implies
		\begin{equation*}
			\dot{b}_{i} = b_{i} \left( \sum_{j=1}^{k} b_{i+j}- \sum_{j=1}^{k} b_{i-j} \right),
			\quad
			i = 1,\dots,N,
		\end{equation*}
		which is \eqref{eq:inbpadd}. The transformation \eqref{eq:MultToAdd} is invertible only when $k$ and $N$ are co-prime, for a more detailed discussion see \cite{Bogoyavlensky1991}.
	\end{remark}

	\subsection{Gibbs Ensemble}
	
	We want to introduce an invariant measure  for the INB lattices
	(\ref{eq:inbpadd}) and \eqref{eq:inbpmul}. Since $H_0 = \prod_{j=1}^N a_j$ and $H_1 = \sum_{j=1}^N a_j$ are constants of motion for all the INB lattices, and both systems are divergence free, in view of \eqref{eq:GGE_formal}  we can consider as invariant measure the same one that we used for the Volterra lattice, namely
	
	\begin{equation} 
		\label{INBGibbs}
		\di\mu_{\INB}\left({\mathbf a;\eta,\beta}\right) = \frac{e^{-\beta H_1 + (\eta-1)\log H_0} \d  \ba }{%
			\int_{\R^N_+}e^{-\beta H_1 + (\eta-1)\log H_0} \d \ba} = 
		\frac{\prod_{j=1}^N a_j^{\eta - 1} e^{-\beta \sum_{j=1}^N a_j} \di \ba }{%
			Z_N^{\INB}(\beta,\eta)}\,, \quad \beta,\eta > 0,
	\end{equation} 
	where the normalization constant $Z_{N}^{\INB}(\beta,\eta)$ has the value given in
	equation \eqref{eq:ZNvolterravalue}.
	For example the random  Lax matrix $L^{(+,k)}$  takes the form
	\vspace{10pt}
		\[
	L^{(+,k)}\simeq\dfrac{1}{2\beta}
	\begin{pmatrix}
	0&1&0&\cdots&\tikz[remember picture]\node[inner sep=0pt] (b) {$\chi^2_{2\eta}$};&0&0&0\\
	0&0&1&\cdots&0&\chi^2_{2\eta}&0&0\\
	0&0&0&1&\cdots&0&\chi^2_{2\eta}&0\\
	\vdots&\ddots&\ddots&\ddots&\ddots&\ddots&&\\
	\chi^2_{2\eta}&0&\cdots&\cdots&0&1&0&\tikz[remember picture]\node[inner sep=0pt] (a) {0};\\
	0&\chi^2_{2\eta}&0\cdots&&\ddots&\ddots&\ddots&\\
	\vdots&\ddots&\ddots&\ddots&\ddots&0&0&1\\
	1&0&\cdots&\chi^2_{2\eta}&\cdots&0&0&0\\
	\end{pmatrix}
	% \ncbar[linewidth=0.6pt, arrows=<->, nodesep=7pt, arrowinset=0.12, linejoin=1,arm=0.25]{T}{B}
	%\node at ({\radius+\xo+0.4}, 2) { \tiny $\Omega_{\beta,r}$};
	 \tikz[overlay]\draw[thick,<->];
	\]
	  \begin{tikzpicture}[overlay, remember picture]
		\draw[stealth-] (a.east)++(1,0) -- node[above] {$k+1$ row }++ (1,0);
		\draw[stealth-] (b.north)++(0,0.1) -- node[right] {$N-k$ column }++ (0,0.5);
	\end{tikzpicture}
	where the $\chi^2_{2\eta}$ distribution has density $\frac{ a^{2\eta-1}}{2^\eta\Gamma(\eta)}e^{-\frac{a}{2}}$.
	Unlike the  Lax matrix of the  Volterra lattice,  the Lax matrices of  these generalizations lack of a known random matrix model to compare with. For this reason,  we present   numerical investigations of  the density of states for these random Lax matrices for several values of the parameters $k, \eta$ and $\beta$,  see Figures \ref{fig:INB_additive1}-\ref{fig:INB_multiplicative1}. We notice that, for both the additive lattice and the multiplicative one, the density of states seems to possess a discrete rotational symmetry. In this spirit, we prove the following
	\begin{lemma}
		Fix $\ell\in \N$. Then for $N$ large enough
		\begin{equation}
			\trace{(L^{(+,k)})^\ell} = \trace{(L^{(\times,k)})^\ell} = 0\,,
		\end{equation}
		if $\ell$ is not an integer multiple of $k+1$.
	\end{lemma}
	
	\begin{proof}
		We prove the statement for the additive case, the proof in the multiplicative one is analogous.

		The main idea is to relate each addendum appearing in $\trace{(L^{(+,k)})^\ell}$ to a specific path in the $\Z^2$ plane, and prove that such a path exists if and only if $\ell = m(k+1)$ for some $m \in \mathbb{N}$.
		In particular, we can focus on the first element of the diagonal of $(L^{(+,k)})^\ell$, write $(L^{(+,k)})^\ell(1,1)$, since all the other ones can be recovered shifting the indices. First, we write $(L^{(+,k)})^\ell(1,1)$ as

		\begin{equation}
			\label{eq:expl_trace}
			(L^{(+,k)})^\ell(1,1) = \sum_{i_1,\ldots,i_{\ell-1} = 1}^N  L^{(+,k)}(1,i_1)L^{(+,k)}(i_1,i_2)\cdots L^{(+,k)}(i_{\ell-1},1)\,.
		\end{equation}
		We notice that, due to the structure of $L^{(+,k)}$, if $L^{(+,k)}(1,i_1)\cdots L^{(+,k)}(i_{\ell-1},1)$ is not zero, then either $i_{s+1} = i_{s}+1$ or $i_{s+1} = i_s -k $ modulo $N$. Now, consider paths in the $\Z^2$ plane from the point $(0,0)$ to $(\ell,0)$, such that the only permitted steps are the up step $(1, 1)$ and  the down step $(1,-k)$. Since these paths resemble the classical Dyck paths, we call them $(1,k)$-Dyck paths of length $\ell$. Given a non-zero element of the product in \eqref{eq:expl_trace}, we can construct the corresponding path in the following way.
		We start at $(0,0)$, then if $|i_1-1| = 1$ we make an up step of height $1$, otherwise we make a down step of height $k$, and so on. 
		
		For each path, let $n$ be the number of up steps and $m$ the number of down steps, then
		
		\begin{equation}
			m + n = \ell\,,\quad n-mk =0\,,
		\end{equation}
		since there is a total of $\ell$ step, and the path has to go back to height $0$. Thus, we deduce that
		\begin{equation}
			m(k+1) = \ell\,,
		\end{equation}
		and the claim is proven.
	\end{proof}
	
	\begin{figure}
		\centering
		\includegraphics[scale = 0.25]{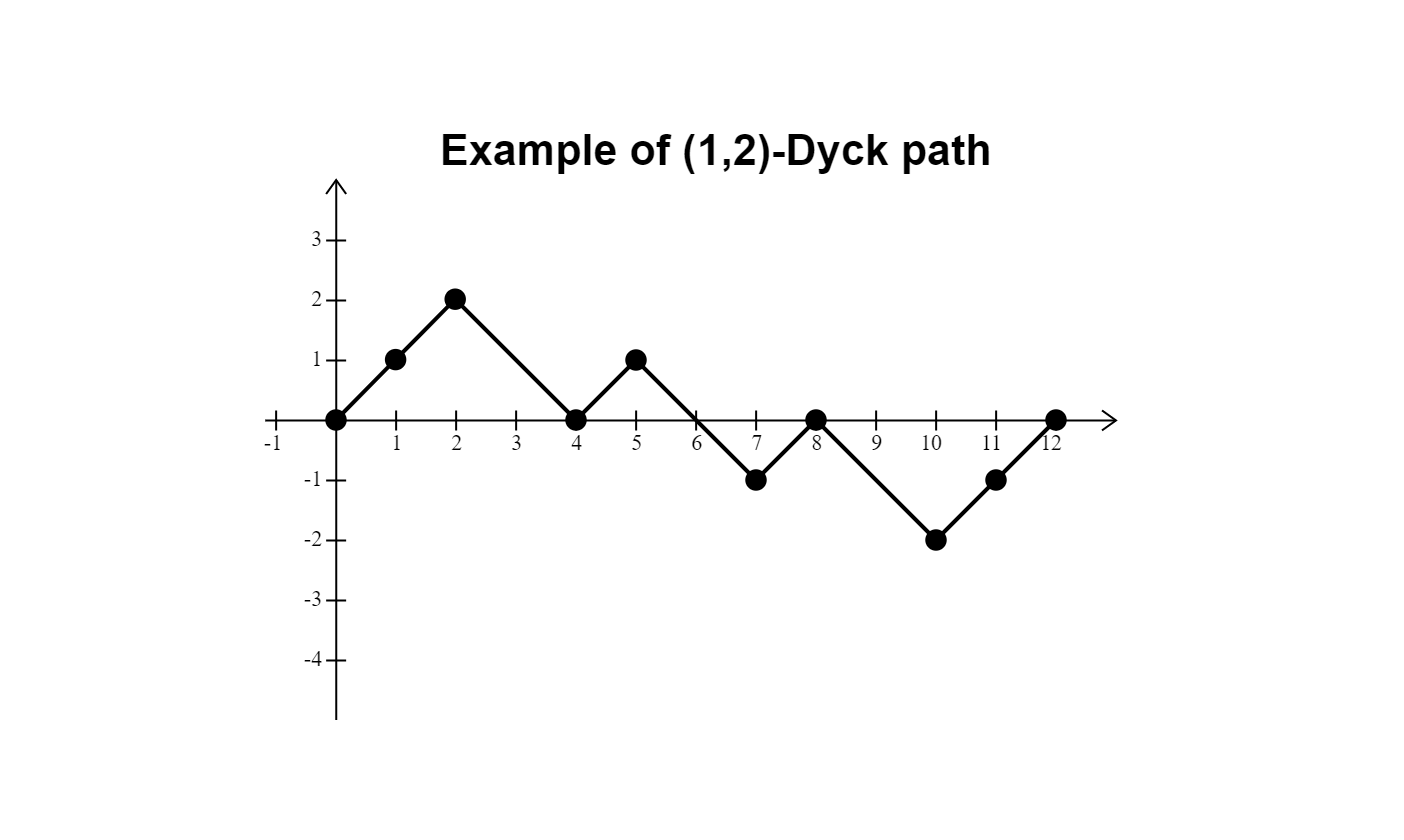}
		\caption{Example of $(1,2)$-Dick path of  length 12}
		\label{fig:path}
	\end{figure}
	
	\begin{remark}
		The previous result implies that the only non-zero moments of the densities of states  $\nu_{\INB,+,k}, \nu_{\INB,\times,k}$, provided they exist, are the ones which are an integer multiple of $k+1$.
	\end{remark}

	\begin{figure}[h]
		\centering
		\includegraphics[scale=0.52]{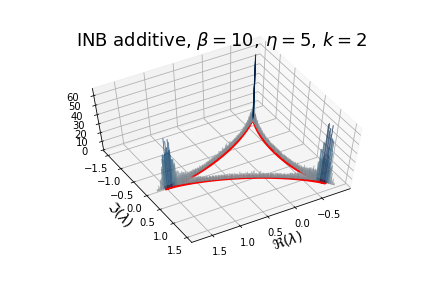}
		\includegraphics[scale=0.52]{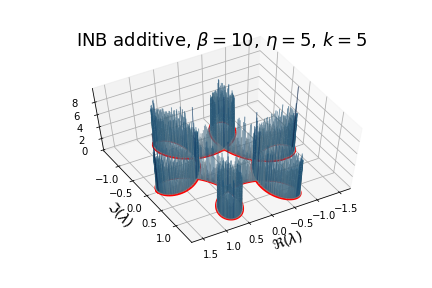}
		
		\includegraphics[scale=0.52]{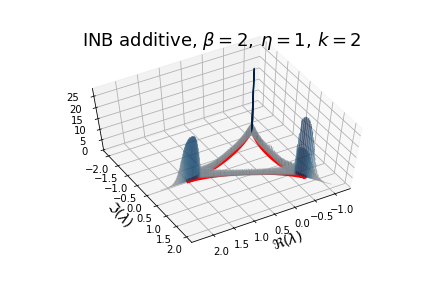}
		\includegraphics[scale=0.52]{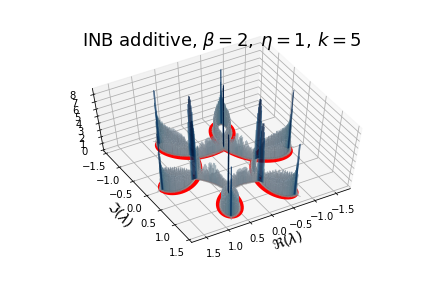}
		\caption{Eigenvalues of INB additive lattice for $k=2$ (left) and $k=5$ (right). $N=1000$ and $6000$ trials performed, in red the corresponding hypotrochoid $\gamma_{+,k}$  defined in equation \eqref{eq:support_INB_guess}. We observe that the examples on the left panel correspond to the case $\wt\eta=\frac{\eta}{\beta}=\frac{1}{k}$ that in the limiting case $\eta,\beta\to\infty$ gives the  arcsin density of states in \eqref{nu_infty}
			where the edges are the cusps of the hypotrochoid. This observation explains the very high peaks located at the cusps.}
		\label{fig:INB_additive1}
	\end{figure}
	
	Another interesting feature of these measures is that their supports seem to be exponentially localized to  one dimensional contours. Specifically, it appears that the supports are the two hypotrochoids $\gamma_{+,k},\, \gamma_{\times,k}$, respectively
	\begin{equation}
		\label{eq:support_INB_guess}
		\gamma_{+,k}(t,\eta,\beta) =  e^{-i t} + \frac{\eta}{\beta} e^{i k t}\,, \qquad \gamma_{\times,k}(t, \eta,\beta) = \frac{\eta}{\beta} e^{ -i t} + e^{i kt}\,,\quad t \in [0;2\pi)\,.
	\end{equation}
	
	This feature is highlighted in figures \ref{fig:INB_additive1}-\ref{fig:INB_multiplicative1}, where we plot the empirical density of states and the corresponding hypotrochoid.
	This characteristic is important since this type of curves are also related to the density of some cyclic digraph, see \cite{Aceituno2019}, and may serve as a link between these two topics.
	
	All these observations lead us to formulate the following conjecture
	
	\begin{conjecture}
		Consider the two matrices $L^{(+,k)},\, L^{(\times,k)}$ as in \eqref{eq:inbpLadd}, \eqref{eq:inbpLmul} both endowed with the probability distribution $\di\mu_{\INB}$ \eqref{INBGibbs}. Then, the densities of states $\nu^{\gamma,\beta}_{\INB,+,k}$ and $ \nu^{\gamma,\beta}_{\INB,\times,k}$ exist, and have a discrete rotational symmetry, namely
		\begin{equation}
			\nu^{\gamma,\beta}_{\INB+,k}(\di z) = \nu^{\gamma,\beta}_{\INB+,k}\left(e^{\frac{2\pi i}{k+1}}\di z\right),\qquad \nu^{\gamma,\beta}_{\INB\times,k}(\di z)= \nu^{\gamma,\beta}_{\INB\times,k}\left(e^{\frac{2\pi i}{k+1}}\di z\right)\,.
		\end{equation}
		Moreover, the densities are exponentially localized in a neighbourhood of the two 
		hypotrochoids $\gamma_{+,k}(t,\eta,\beta) $ and $\gamma_{\times,k}(t,\eta,\beta)$  in \eqref{eq:support_INB_guess} respectively.
	\end{conjecture}

	\begin{figure}
		\centering
		\includegraphics[scale=0.52]{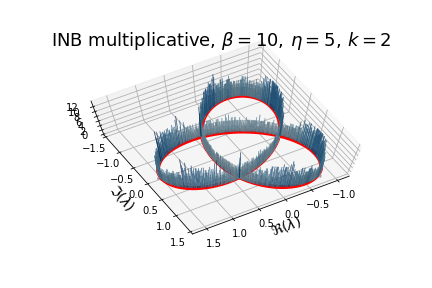}
		\includegraphics[scale=0.52]{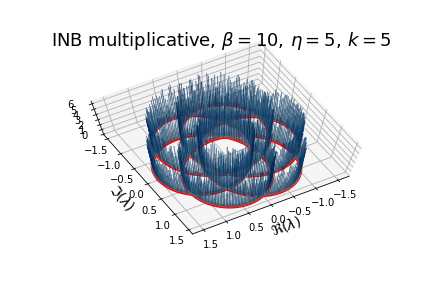}
		
		\includegraphics[scale=0.52]{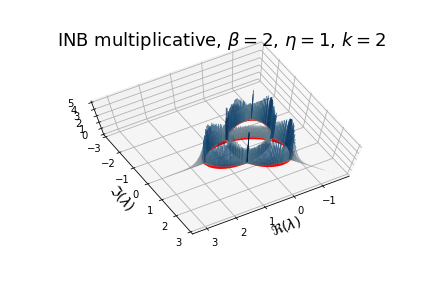}
		\includegraphics[scale=0.52]{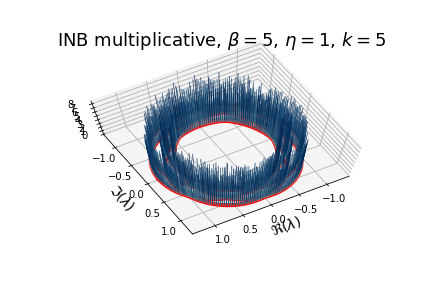}
		\caption{Eigenvalues of INB multiplicative lattice for $k=2$ (left) and $k=5$ (right). $N=1000$ and $6000$ trials performed, in red the corresponding hypotrochoid $\gamma_{\times,k}$}
		\label{fig:INB_multiplicative1}
	\end{figure}
	
	\begin{figure}
		\centering
		\includegraphics[scale=0.52]{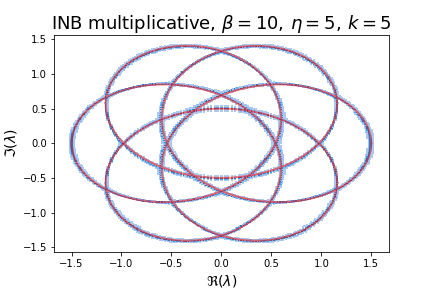}
		\includegraphics[scale=0.52]{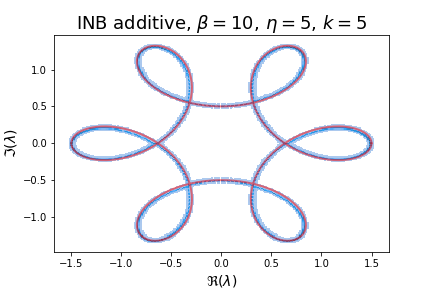}
		\caption{Eigenvalues of INB multiplicative and additive lattice for $k=5$, $N=1000$ and $6000$ trials performed, in red the corresponding hypotrochoid $\gamma_{\times,k}, \gamma_{+,k}$}
	\end{figure}

	\subsection{Parameter limit}
	As in the previous cases, although we are not able to give an explicit formula for the density of states of the INB lattices for general $\beta,\eta$, we can characterize this measure in the low-temperature limit.
	Specifically, we consider the limit as $\beta,\eta\to\infty$ in such a way that $\eta = \wt \eta \beta$, with $\wt\eta$ in a compact set of $\R_+$, and we compute the density of states of the matrices $L^{(+,k)}, L^{(\times,k)}$, endowed with the probability inherited from $\di\mu_{\INB}$ \eqref{INBGibbs}, in this limit. 
	
	The procedure is the same as in the case of Volterra (see Section \ref{SEC:PAR_LIM_VOLT}). Indeed, following the same line, we can conclude that the densities of states $\nu^{\infty}_{\INB,+,k}$ and $\nu^\infty_{\INB,\times, k}$ coincide with the densities of $\fL^{(+)}$ and $\fL^{(\times)}$ respectively, where
	
	\begin{equation}
		\fL^{(+,k)} =  \sum_{i=1}^{N}
		\left( \wt\eta E_{i+k,i} + E_{i,i+1}\right), \quad \fL^{(\times,k)}= \sum_{i=1}^{N}
		\left( E_{i+k,i} + \wt \eta E_{i,i+1}\right)\,.
	\end{equation}
	We notice that both matrices are circulant, thus we can compute their eigenvalues explicitly as
	
	\begin{equation}
		\label{eigenvalues}
		\lambda_j^{(+,k)} = e^{-2\pi i \frac{j}{N}} + \wt \eta e^{2\pi i \frac{jk}{N}}\,, \quad \lambda_j^{(\times,k)} = \wt \eta e^{-2\pi i\frac{j}{N}} + e^{2\pi i \frac{jk}{N}}\,,
	\end{equation}
	here $j = 1, \ldots, N$. Thus, in the large $N$ limit, we deduce that the support of the measures $\nu^{\infty}_{\INB+,k}$ and $\nu^\infty_{\INB\times, k}$   are the hypotrochoids

	\begin{equation}
		\label{eq:support_INB_limit}
		\gamma_{+,k}(t,\wt\eta,1)=  e^{-i t} + \wt \eta e^{ i k t}\,, \quad  \gamma_{\times,k}(t,\wt\eta,1) = \wt \eta e^{- i t} + e^{i kt}\,,\quad t \in [0;2\pi)\,,
	\end{equation}
	and  the limiting eigenvalue densities  are
	\begin{equation}
		\begin{split}
			\label{nu_infty}
			&\nu^{\infty}_{\INB+,k}=       \dfrac{|d z|}{2\pi \sqrt{1+\tilde{\eta}^2k^2-k(| z|^2-1-\tilde{\eta}^2)} },\quad z\in \gamma_{+,k}\,,\\
			&\nu^{\infty}_{\INB\times,k}=\dfrac{|d z|}{2\pi \sqrt{\tilde{\eta}^2+k^2-k(| z|^2-1-\tilde{\eta}^2)} },\quad z\in \gamma_{\times,k}.
		\end{split}
	\end{equation}
	
	%The curves $\gamma_{+,k}(t,\wt \eta), \, \gamma_{\times,k}(t,\wt \eta)$ are, respectively,  the limits of $\phi_{+,k},\, \phi_{\times,k}$ as $\beta \to \infty,\, \eta = \wt \eta \beta$, moreover $\gamma_{+,k}(t,\wt \eta)$ and $\wt \eta\gamma_{\times,k}(t,\frac{1}{\wt \eta})$ have the same support. Examples of these curves are shown in Figure \ref{fig:support}. 
	We summarize these results in the following Proposition.
	\begin{proposition}
		The densities  of states   of the Lax matrices $L^{(+,k)}$ and $ L^{(\times,k)}$ in \eqref{eq:inbpLadd} and \eqref{eq:inbpLmul}  endowed with the Gibbs measure  $\di\mu_{\INB}$ in \eqref{INBGibbs}, in the low temperature limit, i.e.~when $\eta,\,\beta\to\infty$, in such a way that $\eta = \wt \eta \beta$, with $\wt\eta$ in a compact set of $\R_+$,  are  given  respectively  by $\nu^{\infty}_{\INB+,k}$ and $\nu^{\infty}_{\INB\times,k}$ in \eqref{nu_infty}.
		%
		% the  unique measures with support on the curves $\gamma_{+,k}(t),\gamma_{\times,k}(t)$ \eqref{eq:support_INB_limit} such that for any bounded and continuous $f\; : \C \to \R$
		%\begin{equation}
		%    \int_{\gamma_{+,k}} f\nu_{\INB,+} = \int_0^{2\pi}f\left(e^{ -i t} + \wt \eta e^{i k t} \right)\di t\,,\quad \int_{\gamma_{\times,k}} f\nu_{\INB,\times} = \int_0^{2\pi} f\left(\wt \eta e^{-i t} + e^{ i k t} \right)\di t
		%\end{equation}
	\end{proposition}
	
	\begin{remark}
		When the parameters satisfy the relation  $\wt \eta k<1$, the  curve $  \gamma_{+,k}(t,\wt\eta,1)$ is not self-intersecting, while for $\wt \eta k>1$ the curve is 
		self-intersecting. For 
		$\wt \eta k=1$ it  has cusp  singularities
		\cite{DoCarmo2016revised}.
		The limiting shape of the support as $\wt \eta\to0$ is a circle.  \\
		The same considerations are true for the curve $\gamma_{\times,k}(t,\wt \eta)$ upon
		substitution $\wt \eta\mapsto1/\wt\eta$.
		We also observe that the density of states   $\nu^{\infty}_{\INB+,k}$  ( $\nu^{\infty}_{\INB\times,k}$) in equation \eqref{nu_infty} is an Arcsin distribution for  $\wt \eta =\frac{1}{k}$ ($\wt\eta=k$) and the edges correspond to the cusps of the curve  $\gamma_{+,k}(t,\wt \eta)$  ( $\gamma_{\times,k}(t,\wt \eta)$).
		% Furthermore, $\lim_{\wt\eta\to\infty} \gamma_{+,k}(t,\wt \eta)/{\wt\eta} = e^{-ikt}$, which implies that the eigenvalues of $\fL_\infty^{(+,k)} = \lim_{\wt \eta \to \infty} \fL^{(+,k)}/\wt \eta$ lie on the unit circle.
	\end{remark}

	%\begin{figure}
	%    \centering
	%    \includegraphics[scale = 0.5]{INB_additive_measure_support_cusps_k_2.png}
	%    \includegraphics[scale = 0.5]{INB_multiplicative_measure_support_cusps_k_2.png}
	%    
	%    \includegraphics[scale = 0.5]{INB_additive_measure_support_cusps_k_5.png}
	%    \includegraphics[scale = 0.5]{INB_multiplicative_measure_support_cusps_k_5.png}
	%    \caption{Curves $\gamma_+(t)\,, \gamma_\times(t)$\eqref{eq:support_INB_limit} for singular values of the parameters $\wt\eta,k$}
	%    \label{fig:support}
	%\end{figure}
	%We conclude that, in view of the weak convergence of the Gibbs measure \eqref{INBGibbs}, for any continuous and bounded $f\; : \, \C \to \R$, the following limits hold
	%
	%\begin{equation}
	%    \begin{split}
		%         & \frac{1}{N}\sum_{j=1}^N f\left (\lambda_j^{(+,k)}\right) \xrightarrow{N\to \infty} \int_0^{2\pi} f\left(e^{- i t} + \wt \eta e^{i k t} \right)\di t\,, \\ 
		%         &\frac{1}{N}\sum_{j=1}^N f\left(\lambda_j^{(\times,k)}\right) \xrightarrow{N\to \infty} \int_0^{2\pi} f\left(\wt \eta e^{-i t} + e^{ i k t} \right)\di t\,.
		%    \end{split}
	%\end{equation}
	%{In this way, we uniquely characterized  the measures $\nu_{\INB,+}$ and $\nu_{\INB,\times}$.
		\color{black}
		\section{The focusing Ablowitz-Ladik lattice}
		\label{abllad}
		
		The \emph{focusing Ablowitz-Ladik lattice} is the following system of  spatial discrete differential equations 
		\begin{equation}
			\label{eq:AL}
			i	\dot{a}_j +a_{j+1}+a_{j-1}-2a_j+|a_j|^2(a_{j-1}+a_{j+1})=0\,,
		\end{equation}
		where  $a_j\in \C,\, j=1,\ldots,N$, $N\geq 3$, and  we consider periodic boundary conditions $a_{j+N} = a_j$ for all $j\in\Z$. 
		This  equation was introduced by Ablowitz and Ladik \cite{Ablowitz1974,Ablowitz1975}, by  searching integrable  spatial discretization of the cubic  non-linear Schr\"odinger Equation (NLS) for the complex function 
		$\psi(x,t)$, $x\in\R$, $t\in\R^+$
		\begin{equation}
			i \partial_t \psi(x,t) + \partial^2_{x} \psi(x,t) +  2\lvert \psi(x,t) \rvert^2 \psi(x,t)=0\,,
		\end{equation}
		In contrast with what happens in the defocusing case, the particles  $(a_1,\dots, a_N)$  are  free to explore the whole $\C^N$, which is the  phase space of the system. 
		
		On the space  $\cC^\infty(\C^N)$  we consider the Poisson bracket 
		\cite{Ercolani,GekhtmanNenciu}
		\begin{equation}
			\label{eq:poisson_bracket}
			\{f,g\} = i \sum_{j=1}^{N}\rho_j ^2\left(\frac{\partial f}{\partial \wo a_j}\frac{\partial g}{\partial a_j} - \frac{\partial f}{\partial a_j}\frac{\partial g}{\partial \wo a_j}\right)\,, \quad f,g\in\cC^\infty(\C^N).
		\end{equation} 
		
		We notice that the phase shift $a_j(t) \to e^{-2 i t}a_j(t)$ transforms the AL lattice \eqref{eq:AL} into the equation
		\begin{equation}
			\label{AL2}
			\dot{a}_j =i\, \rho_j^2(a_{j+1}+a_{j-1}),\quad  \rho_j=\sqrt{1+|a_j|^2},
		\end{equation}
		which we call the {\it reduced AL  equation}.
		We remark that  the quantity $ H_0 = 2\ln\left(\prod_{j=1}^N \rho_j^2 \right)$ is the generator of the shift $a_j(t)\to e^{-2 i t}a_j(t)$, while $H_1 = - K^{(1)} - \wo{K^{(1)}}$ with 
		\begin{equation}
			\label{K1}
			K^{(1)}:=\sum_{j=1}^{N}a_j\overline{a}_{j+1},
		\end{equation}
		generates the  flow \eqref{AL2}. Therefore, we can rewrite the AL equation as
		\begin{equation}
			\dot a_j = \{a_j, H_{AL} \}\,,\quad H_{AL} = H_0 + H_1\,.
		\end{equation}
		Moreover, it is straightforward to verify that $\{H_0,H_1\} = 0$.
		The Poisson bracket induces the symplectic form 
		\begin{equation}
			\label{eq:symplectic_form}
			\omega = i\sum_{j=1}^{N}\frac{1}{\rho_j^2}\di a_j\wedge\di \wo a_j\,,\quad \rho_j=\sqrt{1+|a_j|^2}\,,
		\end{equation} 
		that is invariant under the evolution generated by the Hamiltonians $H_0$  and $H_1$.
		Therefore, the volume form 
		\[
		\omega^N=\omega\wedge\dots \wedge \omega,
		\]
		is also invariant.
		In view of these properties,  we can define the Gibbs ensemble for the focusing Ablowitz-Ladik lattice on the phase space $\C^N$ as
		\begin{equation}
			\label{eq:Gibbs_AL}
			\di\mu_{AL} = \frac{1}{Z_{N}^{AL}(\beta)} e^{\frac{\beta}{2} H_0} \omega^N 
			=  \frac{1}{Z_{N}^{AL}(\beta)} \prod_{j=1}^N\left( 1+|a_j|^2\right)^{-\beta-1} \di^2 \balpha\,, \quad \beta >0\,,
		\end{equation}
		where $\balpha = (a_1, \ldots, a_N)$,  $ \di^2 \balpha=\prod_{j=1}^N (i\di a_j\wedge \di\overline{a}_j)$ and $Z_{N}^{AL}(\beta)$ is the normalization constant of the system. We notice that according to this measure, all the variables are i.i.d.

		\begin{remark}
			\label{rem:finitemoments}
			The measure with density $\exp(-\beta H_{AL})$ and $\beta >0$ is not bounded nor normalizable on the whole phase space. For this reason, we have defined the Gibbs ensemble as in \eqref{eq:Gibbs_AL}.
			Furthermore, we observe that the measure \eqref{eq:Gibbs_AL} has a finite number of moments, which implies that the corresponding density of states of the Lax matrix (see \eqref{eq:Lax_matrix} below), if it exists, would have a finite number of moments.
		\end{remark}
		
		The focusing AL lattice is a complete integrable system. Indeed it admits a Lax representation, first obtained by Ablowitz and Ladik from the discretization of the Zakharov-Shabat  Lax pair for the focusing non-linear Schr\"odinger equation \cite{ZS}.
		Gesztesy, Holden, Michor, and Teschl  \cite{Teschl2008} found a different Lax pair for the infinite case of  focusing  AL lattice, and for its general hierarchy.
		To adapt their construction, we double the size of the lattice  according to the periodic boundary conditions, thus we consider a chain of $2N$ particles $a_1, \ldots, a_{2N}$ such that $a_j = a_{j+N}$ for $j=1,\ldots, N$.
		Define the  $2\times2$ matrix  $\Xi_j$ 		
		\begin{equation}
			\label{eq:ximatrix}
			\Xi_j = \begin{pmatrix}
				-a_j & \rho_j \\
				\rho_j & -\wo{ a_j}
			\end{pmatrix}\, ,\quad j=1,\dots, 2N\, ,
		\end{equation}
		and the $2N\times 2N$ matrices
		\begin{equation}
			\cM= \begin{pmatrix}
				-\wo{ a_{2N}}&&&&& \rho_{2N} \\
				& \Xi_2 \\
				&& \Xi_4 \\
				&&& \ddots \\
				&&&&\Xi_{2N-2}\\
				\rho_{2N} &&&&&  -a_{2N}
			\end{pmatrix}\, ,\qquad 
			\cL = \begin{pmatrix}
				\Xi_{1} \\
				& \Xi_3 \\
				&& \ddots \\
				&&&\Xi_{2N-1}
			\end{pmatrix} \,.
		\end{equation}	
		Now  let us define the  Lax matrix 
		\begin{equation}
			\label{eq:Lax_matrix}
			\cE  = \cL \cM\,,
		\end{equation}
		that has the structure  of a $5$-band diagonal matrix 
		\[
		\begin{pmatrix}
			*&*&*&&&&&&&*\\
			*&*&*&&&&&&&*\\
			&*&*&*&*&&&&&\\
			&*&*&*&*&&&&&\\
			&&&&&\ddots&\ddots&&&\\
			&&&&&&*&*&*&*\\
			&&&&&&*&*&*&*\\
			*&&&&&&&*&*&*\\
			*&&&&&&&*&*&*\\
		\end{pmatrix}\,.
		\]
		%	The matrix $\cE$ has the same structure of a periodic  CMV  matrix 
		%	 (after Cantero, Morales, and Velásquez  \cite{Cantero2005}), but it is \textbf{not} unitary.
		
		The $N$-periodic equation \eqref{AL2}   is equivalent to  the following Lax equation for the matrix $\cE$,
		\begin{equation}
			\label{eq:Lax_pair}
			\dot \cE = \left[\mathcal{A},\cE\right]\,,
		\end{equation}
		where 
		
		\begin{equation}
			\mathcal{A} = \frac{i}{2}(\cE_+ - \cE_- - \cE_+^{-1} + \cE^{-1}_-)\,,  
		\end{equation}
		where the two projections $M_+,\,M_-$ are defined for a $2N\times 2N$ matrix as 
		\begin{equation}
			\label{eq:proj}
			M_+=\begin{cases}
				M_{\ell,j}\,, \quad \ell < j \leq \ell+ N \\
				M_{\ell,j}\,, \quad \ell > j + N \\
				0 \quad \textrm{otherwise}
			\end{cases}\,, \quad  M_-=\begin{cases}
				M_{\ell,j}\,, \quad j< \ell \leq j+ N \\
				M_{\ell,j}\,, \quad j  > \ell + N \\
				0 \quad \textrm{otherwise}
			\end{cases}\,.
		\end{equation}
		%with
		%\begin{equation}
		%(D_i)_{j,k} = \begin{cases}
			%			\cE_{j,k} 					& k = j + i  \mod \, 2N,\\
			%			0 						& \mbox{otherwise}.
			%		\end{cases}
		%\end{equation}
		%Notice that actually $\wo D_2=D_2$ since it contains instances of $\rho_j \in \mathbb{R}$ only. For example, in the case $N=4$ the matrices $\cE$ and $\cE^+$ read
		%\begin{equation}
		%\begin{pmatrix}
		%SomeMatrix
		%\end{pmatrix}
		%\end{equation}

		We notice that the Lax matrix $\cE$ has a similar structure  to the one of the defocusing AL lattice obtained by Nenciu, and Simon \cite{Nenciu2005,Simon2005}. The crucial difference is that while for  the defocusing AL lattice the blocks $\Xi_j $ are unitary matrices, for the focusing lattice this is not the case since
		$\Xi_j\Xi_j^{\dagger} \neq I_2$ where $I_2=\begin{pmatrix}1&0\\0&1\end{pmatrix}$ and  $^\dagger$ stands for hermitian conjugate.
		
		The measure $\di\mu_{AL}$ induces a probability distribution on the entries of the matrix $\cE$, thus it becomes a random matrix. As in the previous cases, one would like to connect the density of states for this random matrix to the density of states  of some $\beta$-ensemble in the high temperature regime, but, as in the case of the INB lattices, we lack of a matrix representation of  some $\beta$-ensemble with eigenvalues supported on the plane.

		We make the following observations. The matrix $\Xi_j$ \eqref{eq:ximatrix} is complex symmetric, and it can be factorized in the form
		%we notice that it is possible to factorize such matrix in the following way:
		
		\begin{equation}
			\Xi_j = U_j\begin{pmatrix}
				\frac{\wo{a_j}}{|a_j|\left(|a_j| + \sqrt{1+|a_j|^2}\right)} & 0 \\
				0 & - \frac{a_j}{|a_j|}(|a_j| + \sqrt{1+|a_j|^2})\end{pmatrix}U_j\end{equation}  
		where the matrices $U_j = \begin{pmatrix}
			\frac{a_j}{\sqrt{2}|a_j|}& \frac{1}{\sqrt{2}} \\
			\frac{1}{\sqrt{2}} & - \frac{\overline{a_j}}{\sqrt{2}|a_j|}
		\end{pmatrix}$ are unitary, $U_j^{-1} = U_j^{\dagger}$. Thus defining the matrices 
		
		\begin{equation}
			\wt\cM= \begin{pmatrix}
				- \frac{\overline{a_{2N}}}{\sqrt{2}|a_{2N}|}&&&&& \frac{1}{\sqrt{2}} \\
				& U_2 \\
				&& U_4 \\
				&&& \ddots \\
				&&&&U_{2N-2}\\
				\frac{1}{\sqrt{2}} &&&&&  \frac{a_{2N}}{\sqrt{2}|a_{2N}|}
			\end{pmatrix}\, ,\qquad 
			\wt\cL = \begin{pmatrix}
				U_{1} \\
				& U_3 \\
				&& \ddots \\
				&&&U_{2N-1}
			\end{pmatrix} \,,
		\end{equation}	
		we can rewrite the Lax matrix $\cE$ \eqref{eq:Lax_matrix} as 
		
		\begin{equation}
			\label{eq:factorization_AL}
			\cE = \wt\cL \Lambda_{\text{odd}} \wt\cL \wt\cM \Lambda_{\text{even}}\wt\cM\,,
		\end{equation}
		where, defining $c_j:=\frac{\wo{a_j}}{|a_j|\left(|a_j| + \sqrt{1+|a_j|^2}\right)}$, the matrices $\Lambda_{\text{odd}} $ and $ \Lambda_{\text{even}} $ are given by
		\begin{equation}
			\begin{split}
				&\Lambda_{\text{odd}} = \diag\left(c_1,- \frac{1}{c_1},c_3,-\frac{1}{c_3},\dots,c_{2N-1},-\frac{1}{c_{2N-1}}\right)\\
				& \Lambda_{\text{even}} = \diag\left(-\frac{1}{c_{2N}},c_2,-\frac{1}{c_2}, \ldots,c_{2N}\right)\,.
			\end{split}
		\end{equation}
		%\begin{equation}
		%	\begin{split}
			%			&\Lambda_{\text{odd}} = \diag\left(\frac{\wo{a_1}}{|a_1|\left(|a_1| + \sqrt{1+|a_1|^2}\right)},- \frac{a_1}{|a_1|}(|a_1| + \sqrt{1+|a_1|^2}),\frac{\wo{a_3}}{|a_3|\left(|a_3| + \sqrt{1+|a_3|^2}\right)},\ldots\right)\,,  \\ & \Lambda_{\text{even}} = \diag\left(- \frac{a_{2N}}{|a_{2N}|}(|a_{2N}| + \sqrt{1+|a_{2N}|^2}),\frac{\wo{a_2}}{|a_2|\left(|a_2| + \sqrt{1+|a_2|^2}\right)}, \ldots,\frac{\wo{a_{2N}}}{|a_{2N}|\left(|a_{2N}| + \sqrt{1+|a_{2N}|^2}\right)}\right)\,.
			%	\end{split}
		%\end{equation}
		Since we are interested in the distribution of the eigenvalues of $\cE$, it follows from the factorization  \eqref{eq:factorization_AL} that we can also consider the matrix $ \Lambda_{\text{odd}} \wt\cE \Lambda_{\text{even}}\wt\cE^\intercal$, where $\wt\cE = \wt\cL\wt\cM$. The eigenvalues of $\Lambda_{\text{even}},\Lambda_{\text{odd}}$ come in pairs, such that if $\lambda$ is an eigenvalue, then also $-\lambda^{-1}$ is an eigenvalue.  The matrix $\wt\cE$ is a periodic CMV matrix \cite{Cantero2005}, thus its eigenvalues are on the unit circle. 
		
		Thus, we are in a similar  setting considered  in  \cite{WF,WKF,Guionnet2011}. Indeed in \cite{WF} the authors derived the eigenvalues distribution of $U\sqrt{D}$ where $U$ is a Haar distributed unitary matrix and $D$ is a fixed diagonal matrix with positive eigenvalues. They show that the density of states is rotational invariant and it is supported on a  single ring whose radii $r_1<r_2$ satisfy the constraint $d_{min}<r_1<r_2<d_{max}$, where $d_{min}$ and $d_{max}$ are the minimum and maximum eigenvalues of $D$. In \cite{Guionnet2011}, the authors considered a similar problem, namely the characterization of the density of states for a matrix of the form $UTV$, where $U,V$ are independent unitary matrices Haar distributed, and $T$ is a real diagonal matrix independent of $U,V$. They proved, under some mild conditions, that the density of states of the matrix $UTV$ is radially symmetric and it is supported on a ring.
		
		It is therefore reasonable to expect that the density of states of the random Lax matrix of the Ablowitz-Ladik lattice is rotational invariant, but with unbounded support, indeed the eigenvalues of $\Lambda_{\text{even}},\Lambda_{\text{odd}}$ could grow to infinity.
		\begin{figure}
			\centering
			\includegraphics[scale=0.5]{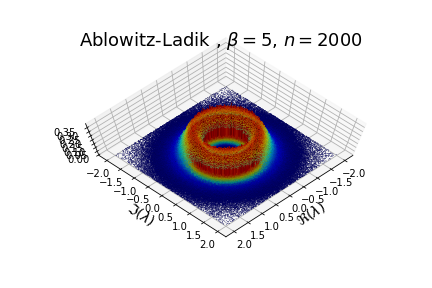}
			\includegraphics[scale=0.5]{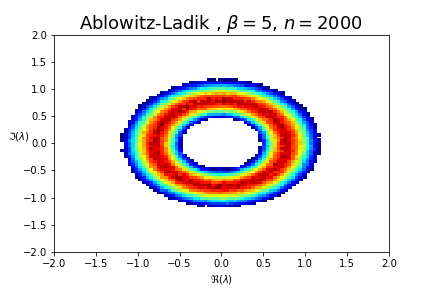}

			\includegraphics[scale=0.5]{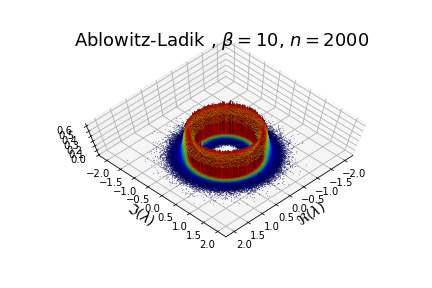}
			\includegraphics[scale=0.5]{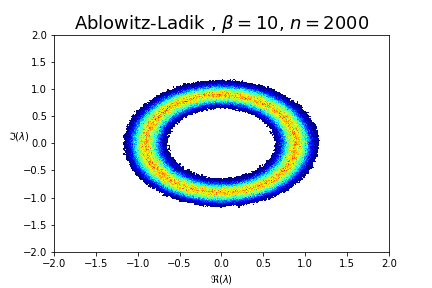}

			\includegraphics[scale=0.5]{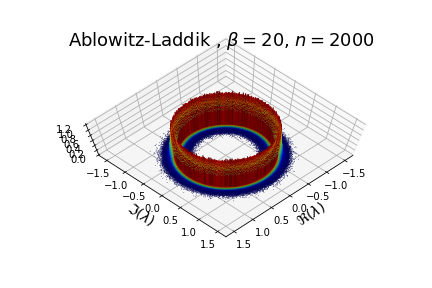}
			\includegraphics[scale=0.5]{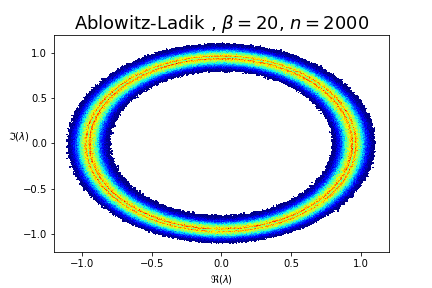}
			
			\caption{Empirical densities for the focusing Ablowitz-Ladik lattice for  $\beta = 5,10,20$, $10000$ trials per picture .}
			\label{fig:eig_density}
		\end{figure}
		
		\begin{figure}
			\centering
			\includegraphics[scale=0.5]{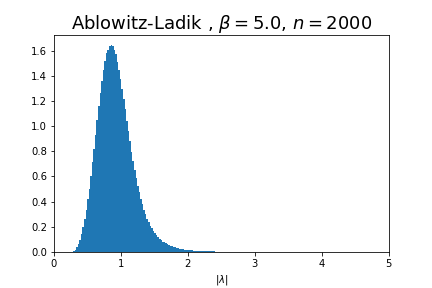}
			\includegraphics[scale=0.5]{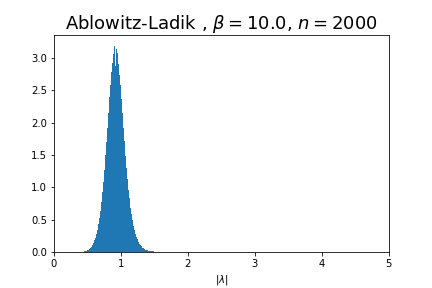}
			
			\includegraphics[scale=0.5]{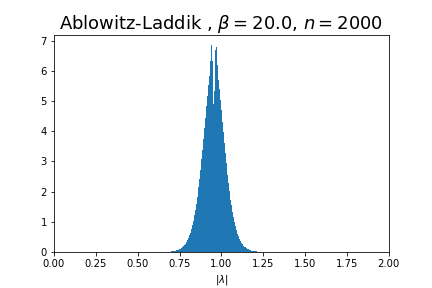}

			\caption{Empirical density for the eigenvalues' modulus for the focusing Ablowitz-Ladik lattice for  $\beta = 5,10,20$, $10000$ trials per picture.}
			\label{fig:modulo_AL}
		\end{figure}
		
		To confirm our expectations, we perform several numerical investigations of the eigenvalues of the random Lax matrix of the Ablowitz-Ladik lattice  for various  values of $\beta$ (see  Fig \ref{fig:eig_density}-\ref{fig:modulo_AL}). In Fig \ref{fig:eig_density} the eigenvalue density is shown. As expected, the density seems to be rotational invariant, and concentrated on a ring, exactly as in \cite{WF,WKF,Guionnet2011}.  For this reason, we investigate the behaviour of the modulus of the eigenvalues, see Fig \ref{fig:modulo_AL}. They seem to be concentrated in a small region, but, in view of Remark \ref{rem:finitemoments}, we expect that the tails should decay just polynomially fast. 
		
		%These considerations lead us to formulate the following conjecture
		%%One immediately notices that for low temperatures, i.e.~for large values of $\beta$, the eigenvalues tend to accumulate on the unit circle; this regime is analysed in the next section. Moreover, the eigenvalue density looks rotational invariant for any value of $\beta$. This last observation leads us to 
		%
		% 
		% \begin{conjecture}
			% Consider the Gibbs ensemble for the focusing Ablowitz-Ladik lattice $ \mu_{AL}$ for $\beta >0$. Then, the density of states of $\cE$ \eqref{eq:Lax_pair} with entries distributed according to $ \mu_{AL}$ is rotational invariant, and is localized in a neighbourhood of a ring.
			% \end{conjecture}
		
		\subsection{Parameter Limit}
		\label{subsec:limit_AL}
		Despite not being able to explicitly compute the density of states for general values of $\beta$, we can perform such an analysis in the low-temperature limit, namely when $\beta \to \infty$. 
		
		We notice that, according to \eqref{eq:Gibbs_AL}, all the $a_j$ are independent. Hence, in order to obtain the density of states in the low-temperature limit, we have to compute the weak limit of the density
		
		\begin{equation}
			\di\mu_\beta = \frac{(1+|z|^2)^{-\beta-1}\di z }{\int_\D (1+|z|^2)^{-\beta-1}\di z}\,.
		\end{equation}
		Proceeding as in the previous cases, it follows that the following limit holds for all bounded and continuous $f\;:  \wo \D \to \R$:
		
		\begin{equation}
			\lim_{\beta\to\infty}\int_\D f(z) \di\mu_\beta = f(0)
		\end{equation}
		%    Therefore the matrix $\cE  =\Lambda_{e}\hat{\cE}\Lambda_{o}$ in \eqref{hatE} becomes a unitary matrix  whose density of states is the uniform distribution on the unit circle.
		The previous limit implies that the density of states of the Ablowitz-Ladik lattice in the low temperature limit is equal to the one of $\wh \cE = \wh \cL\wh \cM$, where $\wh \cL,\,\wh \cM$ are $2N\times 2N$ matrices
		\begin{equation}
			\wh \cM= \begin{pmatrix}
				0&&&&& 1 \\
				& \wt \Xi \\
				&& \wt \Xi \\
				&&& \ddots \\
				&&&&\wt \Xi\\
				1 &&&&&  0
			\end{pmatrix}\, ,\qquad 
			\wh\cL = \begin{pmatrix}
				\wt \Xi \\
				& \wt \Xi \\
				&& \ddots \\
				&&&\wt \Xi
			\end{pmatrix} \,,
		\end{equation}
		and $\wh \Xi$ is defined as the unitary matrix
		\begin{equation}
			\wh \Xi = \begin{pmatrix}
				0& 1 \\
				1 & 0
			\end{pmatrix}\,.
		\end{equation}
		To compute the density of states for the matrix $\cE$, we notice that both $\wh \cM$, and $\wh\cL$ are permutation matrices. Specifically, we identify them with  the following permutations
		\begin{equation}
			\begin{split}
				&\wh \cM \longleftrightarrow (2N,1)(2,3)(4,5)\ldots(2N-2,2N-1)\,,\\
				&\wh \cL \longleftrightarrow (1,2)(3,4)(5,6)\ldots (2N-1,2N)\,.
			\end{split}
		\end{equation}
		As a consequence, the matrix $\wh \cE$ itself corresponds to the permutation
		\begin{equation}
			\wh\cE \longleftrightarrow (1,3,5,\ldots,2N-1)(2,4,6,\ldots,2N)\,.
		\end{equation}
		This implies that the eigenvalues of $\wh \cE$ are all double, and given by
		$$\lambda_j = e^{2\pi i \frac{j}{N}}\,, \quad j =1, \ldots, 2N\,.$$
		From this explicit expression of the eigenvalues, it is straightforward to prove that 
		
		\begin{equation}
			\nu_{AL} = \frac{\di \theta}{2\pi}\,, \quad \theta \in [0,2\pi)\,.
		\end{equation}
		Thus, we have proved the following 
		
		\begin{proposition}
			Consider the Gibbs ensemble $\di\mu_{AL}$ \eqref{eq:Gibbs_AL} of the focusing Ablowitz-Ladik lattice in the low-temperature limit, i.e. $\beta\to\infty$. Then, the density of states $\nu_{AL}$ of the Lax matrix $\cE$ \eqref{eq:Lax_matrix} is given by
			
			\begin{equation}
				\nu_{AL} = \frac{\di \theta}{2\pi}\,, \quad \theta \in [0,2\pi)\,.
			\end{equation}
		\end{proposition}
		
		\section{Schur flow}
		\label{schur}
		
		The \emph{focusing Schur flow}, also known as \emph{discrete mKdV},  is an integrable system deeply related to the Ablowitz-Ladik lattice. Its equations of motion are 
		
		\begin{equation}
			\label{eq:motion_schur}
			\dot a_j = \rho_j^2(a_{j+1} - a_{j-1})\,, \quad \rho_j = \sqrt{1+|a_j|^2}\,, 
		\end{equation}
		where here we consider periodic boundary conditions, $a_{j} = a_{j+N}$ for all $j\in\Z$. Notice that if $a_j(0)\in \R$ for all $j=1,\ldots,N$, then $a_j(t)\in \R$ for all times, implying that $\R^N$ is an invariant subspace for the dynamics.
		%\todo[inline]{Rileggendo tutto lo sviluppo direi che possiamo togliere i vari conigati e moduli}
		
		Recalling the definition \eqref{K1} for $K^{(1)} = \sum_{j=1}^N a_j\wo a_{j+1}$ and introducing the Poisson bracket 
		\begin{equation}
			\{f,g\} = 
			\sum_{j=1}^{N}\rho_j ^2\left(\frac{\partial f}{\partial \wo a_j}\frac{\partial g}{\partial a_j} - \frac{\partial f}{\partial a_j}\frac{\partial g}{\partial \wo a_j}\right)\,, 
			\quad 
			f,g\in\cC^\infty(\C^N),
		\end{equation}
		we can rewrite the equations of motion \eqref{eq:motion_schur} as
		\begin{equation}
			\dot a_j = \{ a_j, H_S\}\,, \quad H_S = K^{(1)} - \wo{K^{(1)}}\,.
		\end{equation}
		
		Notice that the quantity $H_0 = -2\ln\left(\prod_{j=1}^N \rho_j^2\right)$ is a global first integral for the system. Moreover, one can deduce immediately from the equations of motion that $\R^N$ is invariant for the dynamics. Thus, in view of the Hamiltonian representation and this invariance, we define the Gibbs measure for the Schur flow as
		\begin{equation}
			\label{eq:Gibbs_schur}
			\di\mu_S = \frac{1}{Z_{N}^S(\beta)} \prod_{j=1}^N\left( 1+a_j^2\right)^{-\beta-1} \di \balpha\,,\quad a_j\in\R\,,
		\end{equation}
		where $Z_N^S(\beta)$ is the normalization constant of the system,
		\begin{equation}
			Z_N^S(\beta) = \int_{\R^N}\prod_{j=1}^N\left( 1+a_j^2\right)^{-\beta-1} \di \balpha\,.
		\end{equation}

		\begin{remark}
			Similarly to the focusing AL case, the classical Gibbs ensemble is not well-defined on the whole phase space. Indeed, the measure with density $e^{-\beta H_S}$, $\beta >0$ cannot be normalized on $\R^N$. 
		\end{remark}
		%\todo[inline]{Non chiaro: per $\beta>-1$ (con possibile prolungamento 
			%    analitico) si ha:
			%    \begin{equation}
				%        \int_{\R} \frac{\d a}{(1+a^{2})^{\beta+1}}
				%        =
				%        \frac{\sqrt{\pi}}{2} \frac{\Gamma(\beta+1/2)}{\Gamma(\beta+1)}.
				%        \label{eq:betafunc}
				%    \end{equation}
			%}
		
		The Schur flow is a completely integrable system since it admits a Lax formulation. Namely, define the  $2\times2$ matrix  $\Xi_j$ 		
		\begin{equation}
			\Xi_j = \begin{pmatrix}
				-a_j & \rho_j \\
				\rho_j & -\wo{a_j}
			\end{pmatrix}\, ,\quad j=1,\dots, 2N\, ,
		\end{equation}
		and, similarly to the Ablowitz-Ladik case, the $2N\times 2N$ matrices
		\begin{equation}
			\cM= \begin{pmatrix}
				-\wo{a_{2N}}&&&&& \rho_{2N} \\
				& \Xi_2 \\
				&& \Xi_4 \\
				&&& \ddots \\
				&&&&\Xi_{2N-2}\\
				\rho_{2N} &&&&& - a_{2N}
			\end{pmatrix}\, ,\qquad 
			\cL = \begin{pmatrix}
				\Xi_{1} \\
				& \Xi_3 \\
				&& \ddots \\
				&&&\Xi_{2N-1}
			\end{pmatrix} \,.
		\end{equation}
		Then, the $N$-periodic equation \eqref{eq:motion_schur}  is equivalent to  the following Lax equation for the matrix $\cE$:	
		\begin{equation}
			\label{eq:Lax_pair_S}
			\dot \cE = \left[\cA, \cE \right]\,,
		\end{equation}	
		where
		\begin{equation}
			\cA = \frac{1}{2}\left( \cE_+ + \cE_+^{-1} - \cE_- - \cE_-^{-1} \right)\,,
		\end{equation}
		where the two projection $+,-$ are defined in \eqref{eq:proj}.  
		
		Carrying on with the approach of this article, we study the density of states $\nu_S$ for the matrix $\cE$ when the entries are distributed according to the measure  \eqref{eq:Gibbs_schur}. First, we notice that Remark \ref{rem:finitemoments} is valid also in the case of the focusing Schur flow. Moreover, since the variables $a_j$ are real, we can factorize the matrices $\Xi_j$ in the following way:
		{
			\begin{equation}
				\Xi_j =U_{0}        \diag \left(
				c_j , -\frac{1}{c_j}\right)
				U_{0},\,
				\quad c_j = \sqrt{1+a_j^2} - a_j,\,\quad
				U_{0} =
				\frac{1}{\sqrt{2}} 
				\begin{pmatrix}
					1 & 1 \\
					1 & -1
				\end{pmatrix},
			\end{equation}
			where we note that $U_{0}^{-1}=U_{0}$,
			so that the above is a similarity transformation.
			Thus, defining 
			\begin{equation} \label{eq:wtLM}
				\wt\cM= \begin{pmatrix}
					-\frac{1}{\sqrt{2}}&&&&& \frac{1}{\sqrt{2}} \\
					& U_{0} \\
					&& U_{0} \\
					&&& \ddots \\
					&&&&U_{0}\\
					\frac{1}{\sqrt{2}} &&&&&  \frac{1}{\sqrt{2}}
				\end{pmatrix}\, ,\qquad 
				\wt\cL = \begin{pmatrix}
					U_{0} \\
					& U_{0} \\
					&& \ddots \\
					&&&U_{0}
				\end{pmatrix} \,,
			\end{equation}
			we can rewrite the Lax matrix of the Schur flow $\cE$ as}
		\begin{equation}
			\label{eq:factorization_Schur}
			\cE = \wt\cL \Lambda_{\text{odd}} \wt\cL \wt\cM \Lambda_{\text{even}}\wt\cM\,,
		\end{equation}
		where
		\begin{equation}
			\begin{split}
				&\Lambda_{\text{odd}} = \diag\left(c_1, -\frac{1}{c_1},c_3,\ldots\right)\,,  \\ & \Lambda_{\text{even}} = \diag\left(-\frac{1}{c_{2N}},c_2,\ldots, c_{2N}\right)\,.
			\end{split}
		\end{equation}
		As in the case of the Ablowitz-Ladik lattice, since we are interested in just the distribution of the eigenvalues of  $\cE$, we can consider the matrix $ \Lambda_{\text{odd}} \wt\cE \Lambda_{\text{even}}\wt\cE^\intercal$, where 
		\begin{equation}
			\label{EtildeSchur}
			\wt\cE = \wt\cL\wt\cM.
		\end{equation}
		As in the AL case, the eigenvalues of $\Lambda_{\text{even}},\Lambda_{\text{odd}}$ come in pairs, such that if $\lambda$ is an eigenvalue, then also $-\lambda^{-1}$ is an eigenvalue. The main difference with the case of the focusing AL lattice is that in this case the matrix $\wt\cE$ is \emph{deterministic}. Thus, one can be led to think that the eigenvalue distribution of the Schur flow would be similar to the one of the AL lattice, but it is not the case.
		Indeed, we perform several numerical investigations, reported in Fig \ref{fig:eig_schur}, which shows that the behaviour of the eigenvalues is different in the two situations. 
		
		\begin{figure}[ht]
			\centering
			\includegraphics[scale=0.4]{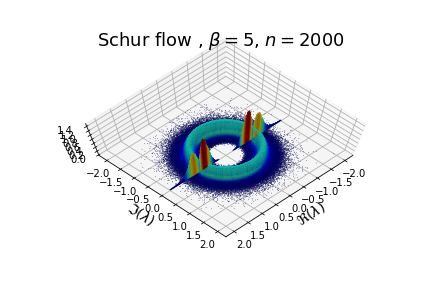}
			\includegraphics[scale=0.4]{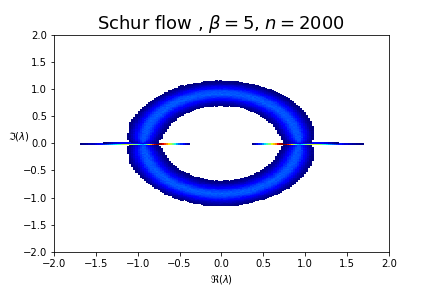}
			
			\includegraphics[scale=0.4]{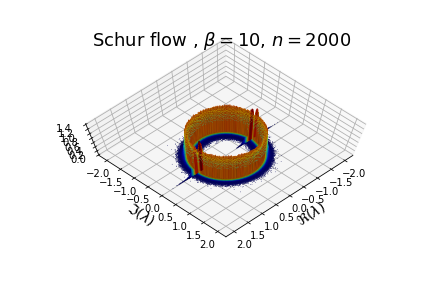}
			\includegraphics[scale=0.4]{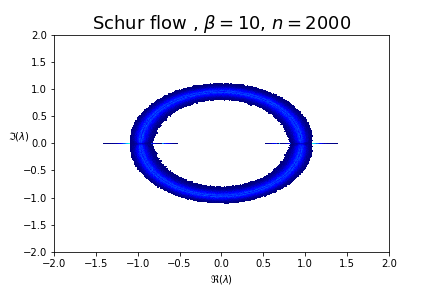}
			
			\includegraphics[scale=0.4]{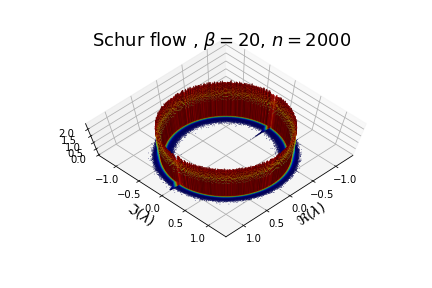}
			\includegraphics[scale=0.4]{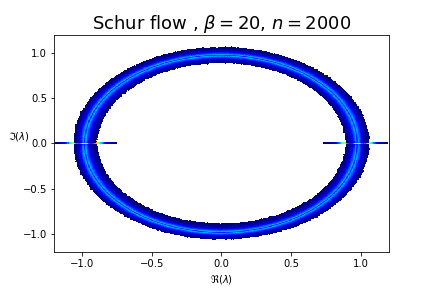}
			
			\caption{Schur flow eigenvalue density for $\beta=5,10,20$, $10000$ trials.}
			\label{fig:eig_schur}
		\end{figure}

		We notice that a consistent part of the eigenvalues tend to stay close to the real axis, see Fig \ref{fig:eig_schur}. This behaviour is also typical of the orthogonal Ginibre ensemble \cite{Edelman1997}. The main reason is that the eigenvalues of $\wt \cE$ are not evenly spaced on the unit circle, but they are constrained to the left semicircle, and are more dense nearby $\pm i$  (see Fig~\ref{fig:constant_matrix}). Indeed we can give an accurate description of the spectrum of this matrix.
		
		More precisely, we have the following.
		
		\begin{proposition} \label{PROP:ESPECTRUM}
			Let $\wt\cE$ be the   $2N \times 2N$ matrix   defined in \eqref{EtildeSchur}.
			Its eigenvalues are the solutions of the quadratic equations
			\begin{equation} \label{eq:wtceEigenvalues}
				\lambda+\frac{1}{\lambda}+1 = \cos\left(\frac{2 \pi j}{N}\right), \quad \quad j=0,1,\dots,N-1,
			\end{equation}
			counting the multiplicity.
		\end{proposition}
		%
		%
		%has eigenvalues
		%\begin{equation}
		%\Lambda\left( \wt\cE \right) =\left\{ \pm i \right \} \cup \left \{ \lambda\ |\ \Omega\left(\lambda \right)^N=\begin{pmatrix} 1 & 0 \\ 0 &1 \end{pmatrix} \right \},
		%\end{equation}
		%where
		%\begin{equation} 
		%\Omega\left(\lambda \right)=\begin{pmatrix} \lambda & -\lambda-1 \\ -\lambda-1 &\frac{\lambda^2+2\lambda+2}{\lambda} \end{pmatrix}.
		%\end{equation}
		%In particular, all the eigenvalues $\lambda$ are solutions of
		%or
		%\begin{proposition}
		%The $2N \times 2N$ matrix $\wt\cE$ has eigenvalues
		%\begin{equation}
		%\Lambda\left( \wt\cE \right) =\left\{ \pm i \right \}  \cup \left \{ \lambda\ |\ \lambda+\frac{1}{\lambda}+1 = \cos\left(\frac{2 \pi j}{N}\right) ,\, j=0,1,\dots,N-1  \right \}
		%\end{equation}
		%\end{proposition}
		\begin{proof}
			See Appendix \ref{app:propProof}.
		\end{proof}
		
		\begin{remark}
			From equation \eqref{eq:wtceEigenvalues} we can infer the limiting distributions of the eigenvalues of $\wt \cE$. We already know all of its eigenvalues lie in the unit circle, hence we can write $\lambda = e^{i \varphi}$ for some $\varphi \in [-\pi,\pi)$. Equation \eqref{eq:wtceEigenvalues} thus becomes
			\begin{equation}
				e^{i \varphi}+e^{-i \varphi}+1=\cos\left(\frac{2 \pi j}{N}\right) \quad \iff \quad \varphi=\arccos\left(\frac 12 \cos\left(\frac{2 \pi j}{N}\right) -\frac12 \right).
			\end{equation}
			Passing to the limit $N\to \infty$, by standard methods, we can compute the limiting density of the argument $\varphi$ of the eigenvalues as
			\begin{equation}
				\mu(\varphi)\d \varphi = \left( \mathbbm{1}_{[\frac{\pi}{2},\pi]}(\varphi) -  \mathbbm{1}_{[-\pi,-\frac{\pi}{2}]}(\varphi) \right) \frac{\sin{\varphi} \,\d\varphi}{\pi \sqrt{1-\left(1+2 \cos{\varphi}\right)^2}}
			\end{equation}
			This behaviour is shown in Fig \ref{fig:constant_matrix}.
		\end{remark}

		\begin{figure}
			\centering
			\includegraphics[scale=0.3]{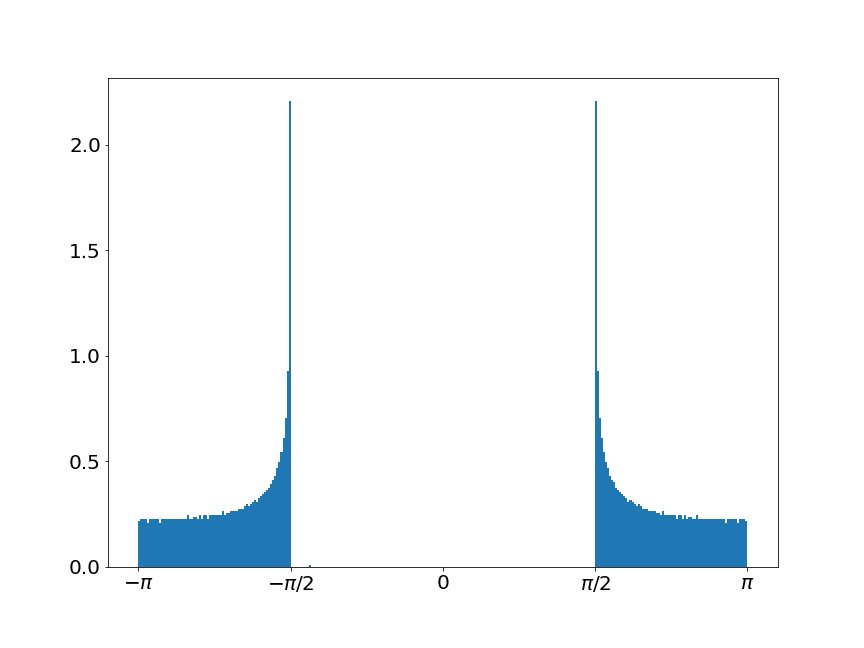}
			\caption{Distribution of the eigenvalues arguments for the $\wt \cE$ \eqref{EtildeSchur}, $ N=5000$.}
			\label{fig:constant_matrix}
		\end{figure} 
		Furthermore, as in the case of the AL lattice, for large $\beta$ the eigenvalues tend to accumulate on the unit circle, see Fig \ref{fig:eig_schur}. In a similar way as done for the focusing AL lattice  we conclude that the density of states of the random Lax matrix $\cE$  in \eqref{eq:Lax_matrix}  with probability distribution entries given by the Gibbs measure
		$\di\mu_{S}$ in  \eqref{eq:Gibbs_schur}, converge in the limit $\beta\to\infty$ to the uniform measure on the unit circle.
		% the next section.   
		%
		%\subsection{Parameter limit}
		%As before, we consider the low-temperature limit for the focusing Schur flow, and we compute the density of states in this regime. 
		%
		%	We notice that, according to \eqref{eq:Gibbs_schur}, all the $a_j$ are independent. So, to obtain the density of states in the ground state, we have to compute the weak limit of the density
		%	
		%	\begin{equation}
			%	     \mu_\beta = \frac{(1+x^2)^{-\beta-1}\di x }{\int_\D (1+x^2)^{-\beta-1}\di x}\,.
			%	\end{equation}
		%As in the AL case, it is straightforward to prove that for any bounded and continuous $f : \C \to \R$,
		%    
		%    \begin{equation}
			%       \lim_{\beta \to \infty} \int f\mu_\beta = f(0)\,,
			%    \end{equation}
		%    which implies that $a_j \rightharpoonup 0$. Hence, we are in a situation analogous to that of subsection \ref{subsec:limit_AL}, from which we deduce the following.
		%    
		%    \begin{lemma}
			%Consider the Gibbs ensemble $\mu_{S}$ \eqref{eq:Gibbs_schur} of the focusing Schur flow in the low-temperature limit, i.e.~$\beta\to\infty$. The density of states $\di\nu_{S}$ of the Lax matrix $\cE$ \eqref{eq:Lax_matrix} is given by
			%\begin{equation}
			%     \nu_{S} = \frac{\di \theta}{2\pi}\,, \quad \theta \in [0,2\pi)\,.
			%\end{equation}
			%\end{lemma}
			%    

			\section{Conclusions and outlook}
			\label{concl}
			
			In this manuscript, we underlined the deep relation between integrable systems and random matrix theory. Specifically, we showed that the exponential Toda lattice, and the Volterra one are related to the Laguerre $\beta$-ensemble at high temperature and the antisymmetric Gaussian $\beta$-ensemble  at high temperature respectively. As we already mentioned in the introduction, these are not isolated and lucky relations, indeed the Toda lattice is related to the Gaussian $\beta$-ensemble at high temperature, and the Ablowitz-Ladik one to the Circular $\beta$ ensemble. Thus, in view of the new relations that we obtained, we conclude that each classical $\beta$-ensemble in the high temperature regime is related to some integrable model, see Table \ref{table:ensemble_integrable}. 
			
			Furthermore, we numerically investigate the additive and multiplicative INB lattices, which are a generalization of the Volterra one. Since the former lattice is related to the antisymmetric Gaussian $\beta$-ensemble, it may be possible that the INB lattice could be related to some generalization of the antisymmetric Gaussian $\beta$-ensemble themselves.  
			
			Finally, we considered the focusing Ablowitz-Ladik lattice and the focusing Schur flow. In view of our numerical investigations, we were able to formulate a precise conjecture regarding the eigenvalues distribution of the focusing Ablowitz-Ladik lattice. Considering the focusing Schur flow, we were not able to formulate a precise conjecture for its eigenvalues distribution.  The expectation is that it would be related to the a possible Ginibre $\beta$-ensemble.
%			 on but we found out that the behaviour of the eigenvalues of this system resembles the one of a Ginibre ensemble. Thus, it may be possible to link these two ensembles, and use the Lax representation of the focusing Schur flow to obtain a matrix representation of the $\beta$ Ginibre ensemble.
			
			It would be fascinating to consider the generalized Gibbs ensemble for the exponential Toda lattice and the Volterra one, and prove that both their empirical measures satisfy a Large Deviation principle in the spirit of \cite{guionnet2021large,mazzuca2022large}. Moreover, it would be interesting to apply the theory of Generalized Hydrodynamic \cite{DoyonNone} to have some insight regarding the behaviour of the correlation functions for these systems.
			\begin{table}[ht]
				\centering
				\begin{tabular}{|c|c|}
					
					\hline
					\textbf{$\beta$-ensemble at high temperature} & \textbf{Integrable System} \\
					\hline
					Gaussian &  Toda lattice \\
					\hline
					Laguerre & Exponential Toda lattice\\
					\hline
					Jacobi & Defocusing Schur flow \\
					\hline
					Circular & Defocusing Ablowitz-Ladik lattice\\
					\hline
					Antisymmetric Gaussian & Volterra lattice \\
					\hline
				\end{tabular}
				\caption{$\beta$-ensembles and integrable systems}
				\label{table:ensemble_integrable}
			\end{table}
			
			\appendix
			
			\section{Proof of Proposition \ref{PROP:ESPECTRUM}} \label{app:propProof}
			
			Recall that the matrix $\wt\cE$ is defined as $\wt\cE = \wt\cL\wt\cM$ where $\wt\cL$ and $\wt\cM$ are as in \eqref{eq:wtLM}. It is a block circulant matrix, indeed we can write
			\begin{equation}
				\wt\cE=\frac12 \begin{pmatrix} E_0 & E_1 & & &E_{-1} \\
					E_{-1} & E_0 & E_1 & &\\
					& \ddots & \ddots & \ddots &\\
					&& \ddots & \ddots &  E_1  \\
					E_1 &&&  E_{-1} &E_0 \end{pmatrix},   
			\end{equation}
			with
			\begin{equation}
				E_0=\begin{pmatrix} -1 & 1 \\ -1 & -1 \end{pmatrix}, \qquad
				E_{-1}=\begin{pmatrix} 0 & 1 \\ 0 & 1 \end{pmatrix}, \qquad
				E_1=\begin{pmatrix} 1 & 0 \\ -1 & 0 \end{pmatrix}.
			\end{equation}
			One can immediately check that $\lambda=\pm i$ are eigenvalues  for $\wt\cE$ with eigenvectors
			\begin{equation}
				v_{\pm i} = \left(\mp i,1, \mp i,1,\dots,\mp i, 1 \right)^\intercal.
				%v_{\pm i} = \left( \mp i, 1, \dots,\mp i, 1 \right)
			\end{equation}
			We now claim that, for fixed $N$, the remaining eigenvalues have multiplicity 2 and are the $(N-1)$ solutions to $\Omega\left(\lambda \right)^N =I_2$, where we defined
			\begin{equation} \label{eq:omega}
				\Omega\left(\lambda \right)=\begin{pmatrix} \lambda & -\lambda-1 \\ -\lambda-1 &\frac{\lambda^2+2\lambda+2}{\lambda} \end{pmatrix}.
			\end{equation}
			Such solutions are obtained by solving the equation
			\begin{equation} \label{eq:omegaPar}
				\lambda + \frac 1\lambda +1 = \cos\left(\frac{2 \pi j}{N}\right)	\qquad for\;\ j=0,\dots,\left\lfloor \frac{N}{2} \right\rfloor\,.
			\end{equation}
			For $j=0$ the solutions to \eqref{eq:omegaPar} are $\pm i$ which we already treated separately. Indeed $\Omega\left(\pm i \right)$ is not diagonalizable and $\Omega\left(\pm i \right)^N \neq I_2$ for every $N$ greater than 0. For any other $j\in\{1,\dots,\left\lfloor \frac{N}{2} \right\rfloor\}$,  the solutions to \eqref{eq:omegaPar} are
			\begin{equation} \label{lambdas}
				\lambda_{1,2} = \frac{\cos\left(\frac{2 \pi j}{N}\right)-1}{2} \pm i \frac{\sqrt{3+2\cos\left(\frac{2 \pi j}{N}\right)-\cos^2\left(\frac{2 \pi j}{N}\right)}}{2}.
			\end{equation}
			Since both the real and imaginary part are monotone functions of  $j$,  different $j's$ will correspond to  different solutions. Hence, if $N$ is odd, we will have a total of $N-1$ solutions coming from \eqref{lambdas}; if $N$ is even one has  $N-2$ distinct solutions 
			coming from the equation in \eqref{eq:omegaPar} plus the double solution $\lambda=-1$ obtained for $j=N/2$.
			
			%As explained above, all of these $(N-1)$ solutions are eigenvalues of multiplicity 2 for $\wt \cE$. Taking into account the particular eigenvalues $\lambda=\pm i$ as well, we have thus found a complete eigensystem for the $2N\times2N$ matrix $\wt \cE$. 

			For a given eigenvalue $\lambda$, the corresponding independent eigenvectors are
			\begin{align}
				v_1 & = \left((1,0) \Omega\left(\lambda \right), \dots,\,(1,0)\Omega\left(\lambda \right)^{N-1},  1, 0\right)^\intercal,	\\
				v_2 &= \left( (0,1)\Omega\left(\lambda \right), \dots,\,(0,1)\Omega\left(\lambda \right)^{N-1},0,1  \right)^\intercal\,.
			\end{align}
			Let us check the correctness of the claim. Write $\wt\cE v_1 :=(w_1, \dots, w_N)^\intercal$,   where $w_j$ are two-dimensional row vectors,
			then using the fact that $\Omega\left(\lambda\right)^N=I_2$, one can compute for any $k=1,\dots,N$,
			\begin{align}
				w_k^\intercal &= \frac12 \left( E_{-1}\, \Omega\left(\lambda \right)^{-1} + E_0 + E_1\,  \Omega\left(\lambda \right) \right)\Omega\left(\lambda \right)^k \begin{pmatrix} 1 \\ 0 \end{pmatrix}	\\
				&= \frac12 \left( \begin{pmatrix} \lambda +1 & \lambda \\ \lambda +1 & \lambda  \end{pmatrix} + \begin{pmatrix} -1 & 1 \\ -1 & -1 \end{pmatrix} + \begin{pmatrix} \lambda  & -\lambda -1 \\ -\lambda  & \lambda +1 \end{pmatrix}\right) \Omega\left(\lambda \right)^k \begin{pmatrix} 1 \\ 0 \end{pmatrix}	\\
				&=\lambda \cdot \Omega\left(\lambda \right)^k \begin{pmatrix} 1 \\ 0 \end{pmatrix},
			\end{align}
			which  shows that $v_1 $ is an eigenvector with eigenvalue $\lambda$. The same proof clearly applies to the other eigenvector $v_2$. 

			\section{Confluent Hypergeometric Functions}
			\label{appendix}
			\subsection{Tricomi Confluent Hypergeometric function}
			
			The Tricomi Confluent Hypergeometric function $\psi(a,b;z)$, that we introduced in Theorem \ref{thm:mazzuca}, is defined as the solution of the Kummer's equation
			
			\begin{equation}
				\label{eq:kummer}
				z \frac{\d^2 \psi}{\d z^2} + (b-z)\frac{\d \psi}{\d z} - a \psi = 0\,,
			\end{equation}
			such that 
			\begin{equation}
				\psi(a,b;z) \sim z^{-a} \,, \quad z\to\infty\,, |\arg(z)| \leq \frac{3}{2}\pi - \delta\,,
			\end{equation}
			here $\delta$ is a small positive constant. The origin is a regular singular point for this equation with indexes $0$ and $1-b$, moreover, this equation has  an irregular singularity at infinity of rank one. Usually, $\psi(a,b;z)$ has a branch point at $0$, whose principal branch is the same as $z^{-a}$, thus it has a cut on the $z$ plane along the interval $(-\infty,0]$.
			
			\subsection{Whittaker function} 
			The Whittaker function $W_{\mu,\kappa}(z)$, that we introduced in Theorem \ref{thm:guido_peter}, is defined as the solution to the Whittaker's equation
			
			\begin{equation}
				\frac{\d^2 W}{\d z^2} + \left(-\frac{1}{4} + \frac{\kappa}{z} + \frac{1-4\mu^2}{4z^2} \right)W = 0\,,
			\end{equation}
			such that
			\begin{equation}
				W_{\mu,\kappa}(z) \sim e^{-\frac{z}{2}}z^{\kappa}\,,   \quad z\to\infty\,, |\arg(z)| \leq \frac{3}{2}\pi - \delta\,,
			\end{equation}
			here $\delta$ is a small positive constant. The Whittaker's equation is obtained from Kummer's equation \eqref{eq:kummer} via the substitution $W = e^{-\frac{z}{2}}z^{\frac{1}{2} + \mu}\psi(a,b;z)$, with $\kappa = \frac{1}{2}(b-a)$, and $\mu = \frac{1}{2}(b-1)$. It has a regular singularity at $0$ with index $\frac{1}{2}\pm \mu$, and an irregular singularity at infinity of rank $1$. Moreover, the following equality holds
			
			\begin{equation}
				W_{\mu,\kappa}(z) = e^{-\frac{z}{2}}z^{\frac{1}{2} +\mu}\psi\left(\frac{1}{2} + \mu -\kappa , 1 + 2\mu; z\right)\,.
			\end{equation}
			
			For a more general overview on Confluent Hypergeometric functions, we refer to \cite[§13]{dlmf}, and \cite[Chapter 13]{Abramowitz1964}.

			\section*{Acknowledgements}
			
			This work was made in the framework of the Project ``Meccanica dei
			Sistemi discreti'' of the GNFM unit of INDAM.
			This material is based upon work supported by the National Science Foundation under Grant No. DMS-1928930 while the authors participated in a program hosted by the Mathematical Sciences Research Institute in Berkeley, California, during the Fall 20-21 semester  "Universality and Integrability in Random Matrix Theory and Interacting Particle Systems".
			M.G.,  T.G. and  G.M.  acknowledge the   Marie Sklodowska–Curie grant No. 778010 IPaDEGAN and the support of GNFM-INDAM group. 
			TG and GM acknowledge the hospitality and support of the   Galileo Galilei Institute,
			programme     "Randomness, Integrability, and Universality".
			
			For the numerical simulations that we presented, we made extensive use of the \texttt{NumPy} \cite{numpy}, \texttt{seaborn} \cite{seaborn}, and \texttt{matplotlib} \cite{matplotlib} libraries.
			
			\section*{Data Availability Statement}
			The python program to generate figures \ref{fig:INB_additive1}--\ref{fig:constant_matrix} can be found at \cite{densityrepo}.
			
			\section*{Conflict of Interest Statement}
			The authors declare no conflict of interest.
			
			\section*{Funding }
			This work was supported bu  Marie Sklodowska–Curie grant No. 778010 IPaDEGAN 
			Moreover, G.M. is financed by the KAM grant number 2018.0344.
			
			\bibliographystyle{siam}
			\bibliography{esplorativo.bib}

		\end{document}